\PassOptionsToPackage{dvipsnames}{xcolor}
\documentclass[
  acmsmall,screen,
]{acmart}
\usepackage{dedsct}

\setcopyright{cc}
\setcctype{by}
\acmDOI{10.1145/3798222}
\acmYear{2026}
\acmJournal{PACMPL}
\acmVolume{10}
\acmNumber{OOPSLA1}
\acmArticle{114}
\acmMonth{4}
\received{2025-10-10}
\received[accepted]{2026-02-17}

\Appendixtrue

\newcommand*\meminitinsecure{
  \begin{jasmincode}[fontsize=\footnotesize,outerwidth=25ex,outerpos=t]
    \jasminindent{0}\\
    \jasminindent{0}\label{line:example:store}%
      [a + \jasminconstant{5}] = sec;\\
    \jasminindent{0}\dots \jasmincomment{// use a}\\
    \jasminindent{0}\label{line:example:begin-while}%
      i = \jasminconstant{0};\\
    \jasminindent{0}\\
    \jasminindent{0}%
      \jasminkw{while} (i < \jasminconstant{10}) \jasminopenbrace{}\\
    \jasminindent{0}\\
      \jasminindent{1}[a + i] = pub;\\
    \jasminindent{1}i += \jasminconstant{1};\\
    \jasminindent{0}\\
    \jasminindent{0}\jasminclosebrace{}\label{line:example:end-while}\\
    \jasminindent{0}\\
    \jasminindent{0}%
      v = [a + \jasminconstant{5}];\label{line:example:load}\\
    \jasminindent{0}\label{line:example:leak}%
      [v] = \jasminconstant{0}; \jasmincomment{// leak v}
  \end{jasmincode}
}

\newcommand*\meminitprotected{
  \begin{jasmincode}[fontsize=\footnotesize,outerwidth=38ex,outerpos=t]
    \jasminindent{0}\label{line:T-example:ms}%
      \msvar{} = \(\bot\);\\
    \jasminindent{0}\label{line:T-example:store}%
      [a + \jasminconstant{5}] = sec;\\
    \jasminindent{0}\dots \jasmincomment{// use a}\\
    \jasminindent{0}i = \jasminconstant{0};\\
    \jasminindent{0}\label{line:T-example:leak}%
      \ileak{\text{(i < \jasminconstant{10})}};\\
    \jasminindent{0}\label{line:T-example:loop}%
      \jasminkw{while} \jasminhighlight{(\hd{\dirvar})} \jasminopenbrace{}\\
    \jasminindent{1}\label{line:T-example:tl-ms1}%
      \dirvar{} = \tl{\dirvar}; \msvar{} ||= !(i < \jasminconstant{10});\\
    \jasminindent{1}[a + i] = pub;\\
    \jasminindent{1}i += \jasminconstant{1};\\
    \jasminindent{1}\ileak{\text{(i < \jasminconstant{10})}};\\
    \jasminindent{0}\jasminclosebrace{}\\
    \jasminindent{0}\label{line:T-example:tl-ms2}%
      \dirvar{} = \tl{\dirvar}; \msvar{} ||= i < \jasminconstant{10};\\
    \jasminindent{0}v = [a + \jasminconstant{5}];\\
    \jasminindent{0}[v] = \jasminconstant{0};
  \end{jasmincode}
}

\newcommand*\meminitspsinsecure{
  \begin{jasmincode}[fontsize=\footnotesize,outerwidth=25ex,outerpos=t]
    \jasminindent{0}\\
    \jasminindent{0}\label{line:slh-example:init-msf}%
      \jasminprimitive{init\_msf}();\\
    \jasminindent{0}[a + \jasminconstant{5}] = sec;\\
    \jasminindent{0}\dots \jasmincomment{// use a}\\
    \jasminindent{0}i = \jasminconstant{0};\\ \\
    \jasminindent{0}%
      \jasminkw{while} (i < \jasminconstant{10}) \jasminopenbrace{}\\
    \jasminindent{0}\\
    \jasminindent{1}[a + i] = pub;\\
    \jasminindent{1}\label{line:slh-example:inc}i += \jasminconstant{1};\\ \\
    \jasminindent{0}\jasminclosebrace{}\\
    \jasminindent{0}\\
    \jasminindent{0}\label{line:slh-example:update-msf}%
      \jasminprimitive{update\_msf}(i == \jasminconstant{10});\\
    \jasminindent{0}v = [a + \jasminconstant{5}];\\
    \jasminindent{0}\label{line:slh-example:protect}%
      v = \jasminprimitive{protect}(v);\\
    \jasminindent{0}\label{line:slh-example:leak}%
      [v] = \jasminconstant{0};
  \end{jasmincode}
}

\newcommand*\meminitspsprotected{
  \begin{jasmincode}[fontsize=\footnotesize,outerwidth=38ex,outerpos=t]
    \jasminindent{0}\msvar{} = \(\bot\);\\
    \jasminindent{0}\label{line:T-slh-example:init-msf}%
      \jasminkw{assert}(!\msvar); \msf{} = \msfnomask{};\\
    \jasminindent{0}[a + \jasminconstant{5}] = sec;\\
    \jasminindent{0}\dots \jasmincomment{// use a}\\
    \jasminindent{0}i = \jasminconstant{0};\\
    \jasminindent{0}\ileak{\text{(i < \jasminconstant{10})}};\\
    \jasminindent{0}%
      \jasminkw{while} \jasminhighlight{(\hd{\dirvar})} \jasminopenbrace{}\\
    \jasminindent{1}%
      \dirvar{} = \tl{\dirvar}; \msvar{} ||= !(i < \jasminconstant{10});\\
    \jasminindent{1}[a + i] = pub;\\
    \jasminindent{1}i += \jasminconstant{1};\\
    \jasminindent{1}\ileak{\text{(i < \jasminconstant{10})}};\\
    \jasminindent{0}\jasminclosebrace{}\\
    \jasminindent{0}%
      \dirvar{} = \tl{\dirvar}; \msvar{} ||= i < \jasminconstant{10};\\
    \jasminindent{0}\label{line:T-slh-example:update-msf}%
      \msf{} ||= !(i == \jasminconstant{10});\\
    \jasminindent{0}v = [a + \jasminconstant{5}];\\
    \jasminindent{0}\label{line:T-slh-example:protect}%
      v = \msf{} ? \jasminconstant{0} : v;\\
    \jasminindent{0}[v] = \jasminconstant{0};
  \end{jasmincode}
}

\newcommand*\meminitTprotected{
  \begin{jasmincode}[fontsize=\footnotesize,outerwidth=40ex,outerpos=t]
    \jasminindent{0}\retvar{} = $\bot$; \obsvar{} = []; \msvar{} = $\bot$;\\
    \jasminindent{0}\jasminkw{if} (\msvar{}) \jasminopenbrace{} \retvar{} = $\top$; \jasminclosebrace{}
      \msf{} = \msfnomask{};\\
    \jasminindent{0}\jasminkw{if} (!\retvar{}) \jasminopenbrace{}\\
    \jasminindent{1}\obsvar{} += [\oaddr{\text{(a + \jasminconstant{5})}}];\\
    \jasminindent{1}[a + \jasminconstant{5}] = sec;\\
    \jasminindent{1}\dots\\
    \jasminindent{1}i = \jasminconstant{0}; \\
    \jasminindent{1}\obsvar{} += [\obranch{\text{(i < \jasminconstant{10})}}];\\
    \jasminindent{1}\jasminkw{while} (\hd{\dirvar}) \jasminopenbrace{}\\
    \jasminindent{2}\dirvar{} = \tl{\dirvar}; \msvar{} ||= !(i < \jasminconstant{10});\\
    \jasminindent{2}\obsvar{} += [\oaddr{\text{(a + i)}}]; \\
    \jasminindent{2}[a + i] = pub;\\
    \jasminindent{2}i += \jasminconstant{1};\\
    \jasminindent{2}\obsvar{} += [\obranch{\text{(i < \jasminconstant{10})}}];\\
    \jasminindent{1}\jasminclosebrace{}\\
    \jasminindent{1}\dirvar{} = \tl{\dirvar}; \msvar{} ||= (i < \jasminconstant{10});\\
    \jasminindent{1}\msf{} ||= !(i == \jasminconstant{10});\\
    \jasminindent{1}\obsvar{} += [\oaddr{\text{(a + \jasminconstant{5})}}];\\
    \jasminindent{1}v = [a + \jasminconstant{5}];\\
    \jasminindent{1}v = \msf{} ? \jasminconstant{0} : v; \\
    \jasminindent{1}\obsvar{} += [\oaddr{\text{v}}];\\
    \jasminindent{1}[v] = \jasminconstant{0};\\
    \jasminindent{0}\jasminclosebrace{}
  \end{jasmincode}
}

\newcommand*\macrotinsecure{
  \begin{jasmincode}[fontsize=\footnotesize,outerwidth=36ex,outerpos=t]
    \jasminindent{0}\\
    \jasminindent{0}ro = sec; i = 0;\\
    \jasminindent{0}\\
    \jasminindent{0}\jasminkw{while} (i < md\_size) \jasminopenbrace{}\\
    \jasminindent{0}\\
    \jasminindent{0}\\
    \jasminindent{1}\label{line:mac-rotation:load}%
      new = [rotated\_mac + ro];\\
    \jasminindent{0}\\
    \jasminindent{1}[out + i] = new;\\
    \jasminindent{1}ro += \jasminconstant{1}; i += \jasminconstant{1};\\
    \jasminindent{1}\label{line:mac-rotation:mod}%
      ro = (md\_size <= ro) ? \jasminconstant{0} : ro;\\
    \jasminindent{1}\\
    \jasminindent{0}\jasminclosebrace{}\\
  \end{jasmincode}
}

\newcommand*\macrottranslated{
  \begin{jasmincode}[fontsize=\footnotesize,outerwidth=48ex,outerpos=t]
    \jasminindent{0}\obsvar{} = []; \msvar{} = \jasminconstant{false};\\
    \jasminindent{0}ro = sec; i = 0;\\
    \jasminindent{0}\obsvar{} = \obsvar{} ++ [\obranch{\text{(i < md\_size)}}];\\
    \jasminindent{0}\jasminkw{while} (\hd{\dirvar}) \jasminopenbrace{}\\
    \jasminindent{1}\dirvar{} = \tl{\dirvar}; \msvar{} ||= !(i < md\_size);\\
    \jasminindent{1}\obsvar{} = \obsvar{} ++ [\oaddr{\text{((rotated\_mac + ro) / 64)}}];\\
    \jasminindent{1}new = [rotated\_mac + ro];\\
    \jasminindent{1}\obsvar{} = \obsvar{} ++ [\oaddr{\text{((out + i) / 64)}}];\\
    \jasminindent{1}[out + i] = new;\\
    \jasminindent{1}ro += \jasminconstant{1}; i += \jasminconstant{1};\\
    \jasminindent{1}\label{line:mac-rotation-sps:mod}%
    ro = (md\_size <= ro) ? \jasminconstant{0} : ro;\\
    \jasminindent{1}\obsvar{} = \obsvar{} ++ [\obranch{\text{(i < md\_size)}}];\\
    \jasminindent{0}\jasminclosebrace{}\\
    \jasminindent{0}\dirvar{} = \tl{\dirvar}; \msvar{} ||= (i < md\_size);
  \end{jasmincode}
}

\newcommand*\caseone{
  \begin{jasmincode}[fontsize=\scriptsize,outerwidth=35ex,outerpos=t]
    \jasminindent{0}\jasmintype{uint8\_t} pub\_mask = \jasminconstant{15};\\
    \jasminindent{0}\jasmintype{uint8\_t} pub[\jasminconstant{16}] = \jasminopenbrace{}\jasminconstant{1}, \dots, \jasminconstant{16}\jasminclosebrace{};\\
    \jasminindent{0}\jasmintype{uint8\_t} pub2[\jasminconstant{512} * \jasminconstant{256}] = \jasminopenbrace{} \jasminconstant{20} \jasminclosebrace{};\\
    \jasminindent{0}\\
    \jasminindent{0}\jasmintype{uint8\_t} sec[\jasminconstant{16}] = \jasminopenbrace{} {\dots} \jasminclosebrace{};\\
    \jasminindent{0}\\
    \jasminindent{0}\jasminstorageclass{volatile} \jasmintype{uint8\_t} temp = \jasminconstant{0};\\
    \jasminindent{0}\jasminstorageclass{volatile} \jasmintype{bool} msf;\\
    \jasminindent{0}\jasmintype{uint8\_t} aux;\\
    \jasminindent{0}\\
    \jasminindent{0}\jasmintype{void} \jasmindname{case\_1}(\jasmintype{uint64\_t} idx) \jasminopenbrace{}\\
    \jasminindent{1}\jasminkw{if} (idx < pub\_size) \jasminopenbrace{}\\
    \jasminindent{2}temp \&= pub2[pub[idx] * \jasminconstant{512}];\\
    \jasminindent{1}\jasminclosebrace{}\\
    \jasminindent{0}\jasminclosebrace{}\\
    \jasminindent{0}\\
    \jasminindent{0}\jasmintype{void} \jasmindname{case\_1\_masked}(\jasmintype{uint64\_t} idx) \jasminopenbrace{}\\
    \jasminindent{1}\label{line:case1:bounds-check}%
      \jasminkw{if} (idx < pub\_size) \jasminopenbrace{}\\
    \jasminindent{2}aux = pub[idx \& pub\_mask];\\
    \jasminindent{2}temp \&= pub2[aux * \jasminconstant{512}];\\
    \jasminindent{1}\jasminclosebrace{}\\
    \jasminindent{0}\jasminclosebrace{}\\
    \jasminindent{0}\\
    \jasminindent{0}\jasmintype{void} \jasmindname{case\_1\_slh}(\jasmintype{uint64\_t} idx) \jasminopenbrace{}\\
    \jasminindent{1}msf = \jasminconstant{false};\\
    \jasminindent{1}\jasminkw{if} (idx < pub\_size) \jasminopenbrace{}\\
    \jasminindent{2}msf = !(idx < pub\_size) | msf;\\
    \jasminindent{2}aux = pub[idx] \& (msf - \jasminconstant{1});\\
    \jasminindent{2}temp \&= pub2[aux * \jasminconstant{512}];\\
    \jasminindent{1}\jasminclosebrace{}\\
    \jasminindent{0}\jasminclosebrace{}
\end{jasmincode}
}

\newcommand*\casefive{
  \begin{jasmincode}[fontsize=\scriptsize,outerwidth=37ex,outerpos=t,startingno=32]
    \jasminindent{0}\jasmintype{void} \jasmindname{case\_5}(\jasmintype{uint64\_t} idx) \jasminopenbrace{}\\
    \jasminindent{1}\jasmintype{int64\_t} i;\\
    \jasminindent{1}\jasminkw{if} (idx < pub\_size) \jasminopenbrace{}\\
    \jasminindent{2}\label{line:case5:loop}%
      \jasminkw{for} (i = idx - \jasminconstant{1}; i >= \jasminconstant{0}; i-{}-) \jasminopenbrace{}\\
    \jasminindent{3}temp \&= pub2[pub[i] * \jasminconstant{512}];\\
    \jasminindent{2}\jasminclosebrace{}\\
    \jasminindent{1}\jasminclosebrace{}\\
    \jasminindent{0}\jasminclosebrace{}\\
    \jasminindent{0}\\
    \jasminindent{0}\jasmintype{void} \jasmindname{case\_5\_masked}(\jasmintype{uint64\_t} idx) \jasminopenbrace{}\\
    \jasminindent{1}\jasmintype{int64\_t} i;\\
    \jasminindent{1}\jasminkw{if} (idx < pub\_size) \jasminopenbrace{}\\
    \jasminindent{2}\jasminkw{for} (i = idx - \jasminconstant{1}; i >= \jasminconstant{0}; i-{}-) \jasminopenbrace{}\\
    \jasminindent{3}aux = pub[i \& pub\_mask];\\
    \jasminindent{3}temp \&= pub2[aux * \jasminconstant{512}];\\
    \jasminindent{2}\jasminclosebrace{}\\
    \jasminindent{1}\jasminclosebrace{}\\
    \jasminindent{0}\jasminclosebrace{}\\
    \jasminindent{0}\\
    \jasminindent{0}\jasmintype{void} \jasmindname{case\_5\_slh}(\jasmintype{uint64\_t} idx) \jasminopenbrace{}\\
    \jasminindent{1}\jasmintype{int64\_t} i;\\
    \jasminindent{1}msf = \jasminconstant{false};\\
    \jasminindent{1}\jasminkw{if} (idx < pub\_size) \jasminopenbrace{}\\
    \jasminindent{2}msf = !(idx < pub\_size) | msf;\\
    \jasminindent{2}\jasminkw{for} (i = idx - \jasminconstant{1}; i >= \jasminconstant{0}; i-{}-) \jasminopenbrace{}\\
    \jasminindent{3}msf = !(i >= \jasminconstant{0}) | msf;\\
    \jasminindent{3}aux = pub[i] \& (msf - \jasminconstant{1});\\
    \jasminindent{3}temp \&= pub2[aux * \jasminconstant{512}];\\
    \jasminindent{2}\jasminclosebrace{}\\
    \jasminindent{1}\jasminclosebrace{}\\
    \jasminindent{0}\jasminclosebrace{}
  \end{jasmincode}
}

\newcommand*\casevoneslh{
  \begin{jasmincode}[fontsize=\footnotesize,outerwidth=23ex]
    \jasminindent{0}\jasmindname{init\_msf}();\\
    \jasminindent{0}s = \jasminconstant{0}; i = \jasminconstant{0};\\
    \jasminindent{0}\jasminkw{while} (i < \jasminconstant{10}) \jasminopenbrace{}\\
    \jasminindent{1}t = p[i];\\
    \jasminindent{1}s += t;\\
    \jasminindent{1}i += \jasminconstant{1};\\
    \jasminindent{0}\jasminclosebrace{}\\
    \jasminindent{0}\jasmindname{update\_msf}(i == 10);\\
    \jasminindent{0}s = \jasmindname{protect}(s);
  \end{jasmincode}
}

\newcommand*\examplevfour{
  \begin{jasmincode}[fontsize=\footnotesize,outerwidth=14ex,outerpos=t]
    \jasminindent{0}\label{line:v4:store-sec}[a] = sec;\\
    \jasminindent{0}\label{line:v4:store-pub}[a] = pub;\\
    \jasminindent{0}\label{line:v4:load}v = [a];\\ \\
    \jasminindent{0}\label{line:v4:leak}[v] = \jasminconstant{0};
  \end{jasmincode}
}

\newcommand*\exampleTvfour{
  \begin{jasmincode}[fontsize=\footnotesize,outerwidth=16ex,outerpos=t]
    \jasminindent{0}[a] = sec;\\
    \jasminindent{0}[a] = pub;\\
    \jasminindent{0}v = [a];\\
    \jasminindent{0}\label{line:Tv4:fence}%
      \jasminprimitive{init\_msf}();\\
    \jasminindent{0}[v] = \jasminconstant{0};
  \end{jasmincode}
}

\newcommand*\labelfigsem{fig:source-semantics}
\newcommand*\labelsemassign{Assign}
\newcommand*\labelsemload{Load}
\newcommand*\labelsemstore{Store}
\newcommand*\labelsemcond{Cond}
\newcommand*\labelsemwhile{While}
\newcommand*\labelseminit{Init}
\newcommand*\labelsemupdate{Update}
\newcommand*\labelsemprotect{Protect}
\newcommand*\labelsemseq{Seq}
\newcommand*\labelsemskip{Skip}
\newcommand*\labelsemrefl{Refl}
\newcommand*\labelsemtrans{Trans}
\newcommand*\labelsembig{Big}

\newcommand*\refsemassign{\refrule{\labelfigsem}{\labelsemassign}}
\newcommand*\refsemload{\refrule{\labelfigsem}{\labelsemload}}
\newcommand*\refsemstore{\refrule{\labelfigsem}{\labelsemstore}}
\newcommand*\refsemcond{\refrule{\labelfigsem}{\labelsemcond}}
\newcommand*\refsemwhile{\refrule{\labelfigsem}{\labelsemwhile}}
\newcommand*\refseminit{\refrule{\labelfigsem}{\labelseminit}}
\newcommand*\refsemupdate{\refrule{\labelfigsem}{\labelsemupdate}}
\newcommand*\refsemprotect{\refrule{\labelfigsem}{\labelsemprotect}}

\newcommand*\refsemrefl{\refrule{\labelfigsem}{\labelsemrefl}}
\newcommand*\refsemtrans{\refrule{\labelfigsem}{\labelsemtrans}}
\newcommand*\refsembig{\refrule{\labelfigsem}{\labelsembig}}

\newcommand*{\figspecsem}{
  \small
  \begin{mathpar}
    \inferrule[\labelsemassign]{
      \vm' = \mset{\vm}{x}{\eval{e}{\vm}}
    }{
      \sem{
        \st{\iassign{x}{e}}{\vm}{\mem}{\ms}
      }{\dstep}{\lnil}{
        \st{\cskip}{\vm'}{\mem}{\ms}
      }
    }

    \inferrule[\labelsemload]{
      i = \eval{e}{\vm}\\
      \vm' = \mset{\vm}{x}{\mem(i)}
    }{
      \sem{
        \st{\iload{x}{e}}{\vm}{\mem}{\ms}
      }{\dstep}{[\oaddr{i}]}{
        \st{\cskip}{\vm'}{\mem}{\ms}
      }
    }

    \inferrule[\labelsemstore]{
      i = \eval{e}{\vm}
    }{
      \sem{
        \st{\istore{e}{x}}{\vm}{\mem}{\ms}
      }{\dstep}{[\oaddr{i}]}{
        \st{\cskip}{\vm}{\mset{\mem}{i}{\vm(x)}}{\ms}
      }
    }

    \inferrule[\labelsemcond]{
      b' = \eval{e}{\vm}
    }{
      \sem{
        \st{\iif{e}{c_\top}{c_\bot}}{\vm}{\mem}{\ms}
      }{[\dforce{b}]}{[\obranch{b'}]}{
        \st{c_{b}}{\vm}{\mem}{\ms \lor (b \neq b')}
      }
    }

    \inferrule[\labelsemwhile]{
      b' = \eval{e}{\vm}\\
      c_\top = c_w \capp \iwhile{e}{c_w}\\
      c_\bot = \cskip
    }{
      \sem{
        \st{\iwhile{e}{c_w}}{\vm}{\mem}{\ms}
      }{[\dforce{b}]}{[\obranch{b'}]}{
        \st{c_{b}}{\vm}{\mem}{\ms \lor (b \neq b')}
      }
    }

    \inferrule[\labelseminit]{
      \vm' = \mset{\vm}{\msf}{\msfnomask}
    }{
      \sem{
        \st{\iinitmsf}{\vm}{\mem}{\msseq}
      }{\dstep}{\lnil}{
        \st{\cskip}{\vm'}{\mem}{\msseq}
      }
    }

    \inferrule[\labelsemupdate]{
      \vm' = \mathite{\eval{e}{\vm}}{\vm}{\mset{\vm}{\msf}{\msfmask}}
    }{
      \sem{
        \st{\iupdatemsf{e}}{\vm}{\mem}{\ms}
      }{\dstep}{\lnil}{
        \st{\cskip}{\vm'}{\mem}{\ms}
      }
    }

    \inferrule[\labelsemprotect]{
      v = \mathite{\vm(\msf) = \msfmask}{0}{\eval{e}{\vm}}
    }{
      \sem{
        \st{\iprotect{x}{e}}{\vm}{\mem}{\ms}
      }{\dstep}{\lnil}{
        \st{\cskip}{\mset{\vm}{x}{v}}{\mem}{\ms}
      }
    }

    \inferrule[\labelsemseq]{
      \sem{
        \st{c}{\vm}{\mem}{\ms}
      }{\seq{d}}{\seq{o}}{
        \st{c'}{\vm'}{\mem'}{\ms'}
      }
    }{
      \sem{
        \st{c \capp c''}{\vm}{\mem}{\ms}
      }{\seq{d}}{\seq{o}}{
        \st{c' \capp c''}{\vm'}{\mem'}{\ms'}
      }
    }

    \inferrule[\labelsemskip]{
    }{
      \sem{
        \st{\cskip \capp c}{\vm}{\mem}{\ms}
      }{\lnil}{\lnil}{
        \st{c}{\vm}{\mem}{\ms}
      }
    }

    \inferrule[\labelsemrefl]{
    }{
      \sem*{s}{\lnil}{\lnil}{s}
    }

    \inferrule[\labelsemtrans]{
      \sem{s}{\seq{d_1}}{\seq{o_1}}{s'}\\
      \sem*{s'}{\seq{d_2}}{\seq{o_2}}{s''}
    }{
      \sem*{s}{\seq{d_1} \lapp \seq{d_2}}{\seq{o_1} \lapp \seq{o_2}}{s''}
    }

    \inferrule[\labelsembig]{
      \sem*{s}{\seq{d}}{\seq{o}}{s'}\\
      \final{s'}
    }{
      \sbigsem{s}{\seq{d}}{\seq{o}}
    }
  \end{mathpar}
}

\newcommand*\labelfigtgtsem{fig:target-semantics}
\newcommand*\labelsemassertT{Assert\(_\top\)}
\newcommand*\labelsemassertF{Assert\(_\bot\)}
\newcommand*\labelsemseqE{Seq\(_\sterror\)}

\newcommand*\labelsemtransE{Trans\(_\sterror\)}

\newcommand*\labelsemtransN{Trans\(_\mathrm{N}\)}

\newcommand*\reftgtsemassign{\refrule{\labelfigtgtsem}{\labelsemassign}}

\newcommand*\reftgtsemassertT{\refrule{\labelfigtgtsem}{\labelsemassertT}}
\newcommand*\reftgtsemassertF{\refrule{\labelfigtgtsem}{\labelsemassertF}}

\newcommand*\reftgtsemrefl{\refrule{\labelfigtgtsem}{\labelsemrefl}}

\newcommand*\reftgtsemtrans{\refrule{\labelfigtgtsem}{\labelsemtrans}}

\newcommand*{\figtgtsem}{
  \small
  \begin{mathpar}
    \inferrule[\labelsemassign]{
    }{
      \sem{
        \st{\iassign{x}{e}}{\vm}{\mem}{}
      }{}{\lnil}{
        \st{\cskip}{\mset{\vm}{x}{\eval{e}{\vm}}}{\mem}{}
      }
    }

    \inferrule[\labelsemload]{
      i = \eval{e}{\vm}
    }{
      \sem{
        \st{\iload{x}{e}}{\vm}{\mem}{}
      }{}{[\oaddr{i}]}{
        \st{\cskip}{\mset{\vm}{x}{\mem(i)}}{\mem}{}
      }
    }

    \inferrule[\labelsemstore]{
      i = \eval{e}{\vm}
    }{
      \sem{
        \st{\istore{e}{x}}{\vm}{\mem}{}
      }{}{[\oaddr{i}]}{
        \st{\cskip}{\vm}{\mset{\mem}{i}{\vm(x)}}{}
      }
    }

    \inferrule[\labelsemcond]{
      b = \eval{e}{\vm}
    }{
      \sem{
        \st{\iif{e}{c_\top}{c_\bot}}{\vm}{\mem}{}
      }{}{[\obranch{b}]}{
        \st{c_b}{\vm}{\mem}{}
      }
    }

    \inferrule[\labelsemwhile]{
      b = \eval{e}{\vm}\\
      c' = \mathite{b}{c_w \capp \iwhile{e}{c_w}}{\cskip}
    }{
      \sem{
        \st{\iwhile{e}{c_w}}{\vm}{\mem}{}
      }{}{[\obranch{b}]}{
        \st{c'}{\vm}{\mem}{}
      }
    }

    \inferrule[\labelsemassertF]{
      \eval{e}{\vm} = \bot
    }{
      \sem{
        \st{\iassert{e}}{\vm}{\mem}{}
      }{}{\lnil}{
        \sterror
      }
    }

    \inferrule[\labelsemassertT]{
      \eval{e}{\vm} = \top
    }{
      \sem{
        \st{\iassert{e}}{\vm}{\mem}{}
      }{}{\lnil}{
        \st{\cskip}{\vm}{\mem}{}
      }
    }

    \inferrule[\labelsemseq]{
      \sem{
        \st{c}{\vm}{\mem}{}
      }{}{\seq{o}}{
        \st{c'}{\vm'}{\mem'}{}
      }
    }{
      \sem{
        \st{c \capp c''}{\vm}{\mem}{}
      }{}{\seq{o}}{
        \st{c' \capp c''}{\vm'}{\mem'}{}
      }
    }

    \inferrule[\labelsemseqE]{
      \sem{\st{c}{\vm}{\mem}{}}{}{\lnil}{\sterror}
    }{
      \sem{
        \st{c \capp c'}{\vm}{\mem}{}
      }{}{\lnil}{
        \sterror
      }
    }

    \inferrule[\labelsemskip]{
    }{
      \sem{
        \st{\cskip \capp c}{\vm}{\mem}{}
      }{}{\lnil}{
        \st{c}{\vm}{\mem}{}
      }
    }

    \inferrule[\labelsemrefl]{
    }{
      \sem*{s}{}{\lnil}{s}
    }

    \inferrule[\labelsemtrans]{
      \sem{s}{}{\seq{o_1}}{s'}\\
      \sem*{s'}{}{\seq{o_2}}{s''_{\sterror}}
    }{
      \sem*{s}{}{\seq{o_1} \lapp \seq{o_2}}{s''_{\sterror}}
    }

    \inferrule[\labelsemtransE]{
      \sem{s}{}{\seq{o_1}}{\sterror}
    }{
      \sem*{s}{}{\seq{o_1}}{\sterror}
    }

    \inferrule[\labelsemtransN]{
    }{
      \nsem{s}{}{\lnil}{s}{0}
    }

    \inferrule[\labelsemtransN]{
      \sem{s}{}{\seq{o_1}}{s'}\\
      \nsem{s'}{}{\seq{o_2}}{s''}{n}
    }{
      \nsem{s}{}{\seq{o_1} \lapp \seq{o_2}}{s''}{n + 1}
    }

    \inferrule[\labelsembig]{
      \sem*{s}{}{\seq{o}}{\st{\cskip}{\vm}{\mem}{}}
      \quad \text{or} \quad
      \sem*{s}{}{\seq{o}}{\sterror}
    }{
      \tbigsem{s}{\seq{o}}
    }
  \end{mathpar}
}

\newcommand*\labelfigssbsem{fig:ssb-semantics}
\newcommand*\refssbsemstore{\refrule{\labelfigssbsem}{\labelsemstore}}
\newcommand*\refssbsemload{\refrule{\labelfigssbsem}{\labelsemload}}
\newcommand*\refssbseminit{\refrule{\labelfigssbsem}{\labelseminit}}

\newcommand*{\figssbsem}{
  \small
  \begin{mathpar}
    \inferrule[\labelsemload]{
      i = \eval{e}{\vm}
      \\
      \vm' = \mset{\vm}{x}{\mem(i)_{n}}
    }{
      \sem{
        \st{\iload{x}{e}}{\vm}{\mem}{\ms}
      }{[\dload{n}]}{[\oaddr{i}]}{
        \st{\cskip}{\vm'}{\mem}{\ms \lor (0 \neq n)}
      }
    }

    \inferrule[\labelsemstore]{
      i = \eval{e}{\vm}
      \\
      \mem' = \mset{\mem}{i}{\vm(x) \lapp \mem(a)}
    }{
      \sem{
        \st{\istore{e}{x}}{\vm}{\mem}{\ms}
      }{\dstep}{[\oaddr{i}]}{
        \st{\cskip}{\vm}{\mem'}{\ms}
      }
    }

   \inferrule[\labelseminit]{
     \vm' = \mset{\vm}{\msf}{\msfnomask}
     \\
     \qA{i}{\mem'(i) = \mem(i)_{0}}
   }{
     \sem{
       \st{\iinitmsf}{\vm}{\mem}{\msseq}
     }{\dstep}{\lnil}{
       \st{\cskip}{\vm'}{\mem'}{\msseq}
     }
   }
  \end{mathpar}
}

\newcommand*{\figtrans}{%
  \small
  \begin{mathpar}
    \renewcommand*\arraystretch{1.3}
    \begin{array}[t]{r @{{}\eqdef{}} l}
      \Tprogstar{\iassign{x}{e}} & \iassign{x}{e}\\
      \Tprogstar{\iload{x}{e}} & \iload{x}{e}\\
      \Tprogstar{\istore{e}{x}} & \istore{e}{x}\\
      \Tprogstar{\cnil} & \cnil\\
      \Tprogstar{\iif{e}{c_\top}{c_\bot}} &
        \renewcommand*\arraystretch{1}
        \begin{array}[t]{@{}l}
          \ileak{e}\capp \\
          \mathbf{if}~(\hd{\dirvar}) \{ \\
          \quad \iassign{\dirvar}{\tl{\dirvar}}\capp \\
          \quad \iassign{\msvar}{\msvar\lor \neg~e}\capp \\
          \quad \Tprogstar{c_\top} \\
          \}~\mathbf{else}~\{ \\
          \quad \iassign{\dirvar}{\tl{\dirvar}}\capp \\
          \quad \iassign{\msvar}{\msvar\lor e}\capp \\
          \quad \Tprogstar{c_\bot} \\
          \}
        \end{array}
    \end{array}

    \begin{array}[t]{r @{{}\eqdef{}} l}
      \Tprogstar{\iwhile{e}{c}} &
        \renewcommand*\arraystretch{1}
        \begin{array}[t]{@{}l}
          \ileak{e}\capp \\
          \mathbf{while}~(\hd{\dirvar})~\{\\
          \quad \iassign{\dirvar}{\tl{\dirvar}}\capp  \\
          \quad \iassign{\msvar}{\msvar\lor \neg e}\capp \\
          \quad \Tprogstar{c}\capp \\
          \quad \ileak{e} \\
          \}\capp \\
          \iassign{\dirvar}{\tl{\dirvar}}\capp \\
          \iassign{\msvar}{\msvar \lor e}
        \end{array}\\
      \Tprogstar{\iinitmsf{}} & \iassert{\neg \msvar} \capp \iassign{\msf}{\msfnomask}\\
      \Tprogstar{\iupdatemsf{e}} & \iassign{\msf}{\msf\lor \neg~e}\\
      \Tprogstar{\iprotect{x}{e}} & \iassign{x}{\eite{\msf}{0}{e}}\\
      \Tprogstar{c \capp c'} & \Tprogstar{c} \capp \Tprogstar{c'}
    \end{array}

    \Tprog{c} \eqdef \iassign{\msvar}{\bot} \capp \Tprogstar{c}
\end{mathpar}
}

\newcommand*{\figssbtrans}{
  \small\renewcommand*\arraystretch{1.3}%
  \[\begin{array}{r @{{}\eqdef{}} l}
    \Tssbstar{\iload{x}{e}} &
      \begin{array}[t]{@{}l}
        \iassign{x}{\nth{[e]}{\hd{\dirvar}}}\capp\\
        \iassign{\msvar}{\msvar \lor (0 \neq \hd{\dirvar})}\capp\\
        \iassign{\dirvar}{\tl{\dirvar}}
      \end{array}
    \\
    \Tssbstar{\istore{e}{x}} & \istore{e}{\app{[e]}{x}}
    \\
    \Tssbstar{\iinitmsf} &
      \renewcommand*\arraystretch{1}
      \begin{array}[t]{@{}l}
        \iassert{\neg \msvar}\capp \iassign{\msf}{\msfnomask}\capp\iclearmem{}
      \end{array}
  \end{array}\]
}

\newcommand*{\figasserttrans}{
  \small\renewcommand*\arraystretch{1.3}%
  \[\begin{array}{r @{{}\eqdef{}} l}
      \Tassstar{\cnil} & \cnil
      \\
      \Tassstar{\iassert{e}} & \iift{\neg e}{\iassign{\retvar}{\top}}
      \\
      \Tassstar{\iif{e}{c_\top}{c_\bot}} & \iif{e}{\Tassstar{c_\top}}{\Tassstar{c_\bot}}
      \\
      \Tassstar{\iwhile{e}{c_{body}} \capp c} &
        \renewcommand*\arraystretch{1.3}
        \begin{cases}
          \iwhile{e}{\Tassstar{c_{body}}}
          & \text{if } \Tassstar{c_{body}} = c_{body} \text,
          \\
          \iwhile{e \land \neg\retvar}{\Tassstar{c_{body}}}
          & \text{otherwise.}
        \end{cases}
      \\
      \Tassstar{c \capp c'} &
        \begin{cases}
          \Tassstar{c} \capp \Tassstar{c'} & \text{if } \Tassstar{c} = c \text,
          \\
          \Tassstar{c} \capp \iift{\neg\retvar}{\Tassstar{c}} & \text{otherwise.}
        \end{cases}
      \\
      \Tass{c} & \iassign{\retvar}{\bot} \capp \Tassstar{c}
  \end{array}\]
}

\newcommand*{\figleaktrans}{
  \small
  \begin{mathpar}
    \renewcommand{\arraystretch}{1.3}
    \begin{array}[t]{r @{{}\eqdef{}} l}
      \Tleakstar{\iassign{x}{e}} & \iassign{x}{e}
      \\
      \Tleakstar{\iload{x}{e}}
      & \renewcommand{\arraystretch}{1}
      \begin{array}[t]{@{}l}
        \iappend{\obsvar}{\oaddr{e}} \capp\\
        \iload{x}{e}
      \end{array}
      \\
      \Tleakstar{\istore{x}{e}}
      & \renewcommand{\arraystretch}{1}
      \begin{array}[t]{@{}l}
        \iappend{\obsvar}{\oaddr{e}} \capp\\
        \istore{e}{x}
      \end{array}
      \\
      \Tleakstar{\cnil} & \cnil
      \\
      \Tleakstar{c \capp c'} & \Tleakstar{c} \capp \Tleakstar{c'}
    \end{array}

    \begin{array}[t]{r @{{}\eqdef{}} l}
      \Tleakstar{\iif{e}{c}{c'}}
      & \renewcommand{\arraystretch}{1}
      \begin{array}[t]{@{}l}
        \iappend{\obsvar}{\obranch{e}} \capp\\
        \iif{e}{\Tleakstar{c}}{\Tleakstar{c'}}
      \end{array}
      \\
      \Tleakstar{\iwhile{e}{c}}
      & \renewcommand{\arraystretch}{1}
      \begin{array}[t]{@{}l}
        \iappend{\obsvar}{\obranch{e}} \capp\\
        \mathbf{while}~(e)~\{\\
        \quad \Tleakstar{c} \capp\\
        \quad \iappend{\obsvar}{\obranch{e}}\\
        \}
      \end{array}
    \end{array}

    \Tleak{c} \eqdef{} \iassign{\obsvar}{\enil} \capp \Tleakstar{c}
  \end{mathpar}
}

\newcommand*\figproducttrans{
  \small
  \begin{mathpar}
    \renewcommand{\arraystretch}{1.3}
    \begin{array}[t]{r @{{}\eqdef{}} l}
      \Tprod*{\iassign{x}{e}}
      & \iassign{\rtag{x}{1}}{\rtag{e}{1}} \capp
        \iassign{\rtag{x}{2}}{\rtag{e}{2}}
      \\
      \Tprod*{\iload{x}{e}}
      & \renewcommand{\arraystretch}{1}
      \begin{array}[t]{@{}l}
        \iassert{\rtag{e}{1} = \rtag{e}{2}} \capp\\
        \iload{\rtag{x}{1}}{\rtag{e}{1}} \capp\\
        \iload{\rtag{x}{2}}{\rtag{e}{2}}
      \end{array}
      \\
      \Tprod*{\istore{e}{x}}
      & \renewcommand{\arraystretch}{1}
      \begin{array}[t]{@{}l}
        \iassert{\rtag{e}{1} = \rtag{e}{2}} \capp\\
        \istore{\rtag{e}{1}}{\rtag{x}{1}} \capp\\
        \istore{\rtag{e}{2}}{\rtag{x}{2}}
      \end{array}
      \\
      \Tprod*{\cnil} & \cnil
    \end{array}

    \begin{array}[t]{r @{{}\eqdef{}} l}
      \Tprod*{\iif{e}{c}{c'}}
      & \renewcommand{\arraystretch}{1}
      \begin{array}[t]{@{}l}
        \iassert{\rtag{e}{1} = \rtag{e}{2}} \capp\\
        \iif{\rtag{e}{1}}{\Tprod*{c}}{\Tprod*{c'}}
      \end{array}
      \\
      \Tprod*{\iwhile{e}{c}}
      & \renewcommand{\arraystretch}{1}
      \begin{array}[t]{@{}l}
        \iassert{\rtag{e}{1} = \rtag{e}{2}} \capp\\
        \mathbf{while}~(\rtag{e}{1})~\{\\
        \quad \Tprod*{c} \capp\\
        \quad \iassert{\rtag{e}{1} = \rtag{e}{2}}\\
        \}
      \end{array}
      \\
      \Tprod*{c \capp c'} & \Tprod*{c} \capp \Tprod*{c'}
    \end{array}

    \Tprod{c}{\phi} \eqdef \iift{\phi}{\Tprod*{c}}
  \end{mathpar}
}

\newcommand*\tabspsbinsec{
  \small
  \begin{tabular}{l l r r r}
    \toprule
    Variant
    & Kind
    & \multicolumn{1}{l}{SPS-\binsec{}}
    & \multicolumn{1}{l}{\binsechaunted{}}
    & \multicolumn{1}{l}{SPS-\ctgrind{}}
    \\
    \midrule
    Vulnerable
    & Loop-Free
    & 12\evaltimes{} \evalbug{}
    & 12\evaltimes{} \evalbug{}
    & 12\evaltimes{} \evalbug{}
    \\
    & Loop
    & 4\evaltimes{} \evalbug{}
    & 4\evaltimes{} \evalbug{}
    & 4\evaltimes{} \evalbug{}
    \\[1ex]
    Index-masked
    & Loop-Free
    & 12\evaltimes{} \evalok{}
    & 12\evaltimes{} \evalok{}
    & 12\evaltimes{} \evalabort{}
    \\
    & Loop
    & 4\evaltimes{} \evalabort{}
    & 4\evaltimes{} \evalabort{}
    & 4\evaltimes{} \evalabort{}
    \\[1ex]
    SelSLH-protected
    & Loop-Free
    & 12\evaltimes{} \evalok{}
    & 12\evaltimes{} \evalok{}
    & 12\evaltimes{} \evalabort{}
    \\
    & Loop
    & 4\evaltimes{} \evalabort{}
    & 4\evaltimes{} \evalabort{}
    & 4\evaltimes{} \evalabort{}
    \\
    \bottomrule
  \end{tabular}
}

\begin{document}

\title{(Dis)Proving Spectre Security with Speculation-Passing Style}

\author{Santiago Arranz-Olmos}
\orcid{0009-0007-7425-570X}
\affiliation{%
  \institution{MPI-SP}
  \city{Bochum}
  \country{Germany}
}
\email{santiago.arranz-olmos@mpi-sp.org}

\author{Gilles Barthe}
\orcid{0000-0002-3853-1777}
\affiliation{%
  \institution{MPI-SP}
  \city{Bochum}
  \country{Germany}
}
\affiliation{%
  \institution{IMDEA Software Institute}
  \city{Madrid}
  \country{Spain}
}
\email{gilles.barthe@mpi-sp.org}

\author{Lionel Blatter}
\orcid{0000-0001-9058-2005}
\affiliation{%
  \institution{MPI-SP}
  \city{Bochum}
  \country{Germany}
}
\email{lionel.blatter@mpi-sp.org}

\author{Xingyu Xie}
\orcid{0000-0002-2220-7294}
\affiliation{%
  \institution{MPI-SP}
  \city{Bochum}
  \country{Germany}
}
\email{xingyu.xie@mpi-sp.org}

\author{Zhiyuan Zhang}
\orcid{0009-0000-2669-5654}
\affiliation{%
  \institution{MPI-SP}
  \city{Bochum}
  \country{Germany}
}
\email{zhiyuan.zhang@mpi-sp.org}



\begin{abstract}
  Constant-time (CT) verification tools are commonly used for
  detecting potential side-channel vulnerabilities in cryptographic
  libraries.  Recently, a new class of tools, called speculative
  constant-time (SCT) tools, has also been used for detecting
  potential Spectre vulnerabilities. In many cases, these SCT tools
  have emerged as liftings of CT tools. However, these liftings are
  seldom defined precisely and are almost never analyzed formally. The
  goal of this paper is to address this gap, by developing formal
  foundations for these liftings, and to demonstrate that these
  foundations can yield practical benefits.

  Concretely, we introduce a program transformation, coined
  Speculation-Passing Style (SPS), for reducing SCT verification to CT
  verification. Essentially, the transformation instruments the
  program with a new input that corresponds to attacker-controlled
  predictions and modifies the program to follow them. This approach
  is sound and complete, in the sense that a program is
  SCT if and only if its SPS transform is CT\@. Thus, we can leverage
  existing CT verification tools to prove SCT\@; we illustrate this by
  combining SPS with three standard methodologies for CT verification,
  namely reducing it to noninterference, assertion safety, and dynamic
  taint analysis. We realize these combinations with three existing tools,
  \easycrypt{}, \binsecrel{}, and \ctgrind{}, and we evaluate them on Kocher's
  benchmarks for Spectre-v1. Our results focus on Spectre-v1 in the
  standard CT leakage model; however, we also discuss applications of
  our method to other variants of Spectre and other leakage models.
\end{abstract}


\maketitle

\section{Introduction}%
\label{sec:intro}
The constant-time programming discipline is a gold standard for
cryptographic libraries~\cite{CT_Survey,Geimer2023} that protects
software against timing- and cache-based side-channel attacks.
Such attacks can invalidate all security guarantees of critical
cryptographic implementations~\cite{DBLP:conf/sp/AlFardanP13,AlbrechtP16,%
  DBLP:journals/joc/TromerOS10,DBLP:conf/ccs/GenkinVY17,%
  DBLP:journals/jce/YaromGH17,DBLP:conf/ccs/RonenPS18,%
  DBLP:conf/ccs/AranhaN0TY20}.
Unfortunately, writing constant-time code is extremely difficult, even
for experts.
The challenges of writing CT code, and the risks of deploying code that
is not CT, have stirred the development of CT analysis tools,
which can help programmers ensure that their code is CT~\cite{%
CT_Survey,Geimer2023,%
  DBLP:conf/sp/BarbosaBBBCLP21,DBLP:conf/sp/JancarFBSSBFA22}.

Despite its success in protecting cryptography from side-channel
attacks, the CT discipline emerged more than twenty years ago and
offers no protection against Spectre attacks~\cite{spectre}. These
attacks exploit speculative execution (see,
e.g.,~\cite{DBLP:conf/uss/CanellaB0LBOPEG19}), rendering CT mitigations
ineffective because they fall outside of the threat model of CT, which
is based on a sequential execution model.
Speculative constant-time~\cite{pitchfork,highassuranceSpectre}
remedies this gap by lifting the principles of CT to a threat model
based on a speculative model of execution, and has already been adopted
by several post-quantum cryptography implementations.
The drastically stronger threat model of SCT makes writing SCT code even
harder than writing constant-time code.
To aid programmers in writing SCT code, numerous SCT analysis tools have
emerged over the last five years~\cite{SCT_Survey}. These efforts have
given programmers access to a range of complementary tools based on
different techniques such as dynamic analysis, symbolic execution,
static analysis, and type systems.

\paragraph{Research Questions}
Many SCT verification tools have been conceived as liftings of CT
verification tools. However, these liftings are generally described
informally, and used to inform the design of the new tool. They are
very seldom formalized, and almost never studied formally---a notable
exception is~\cite{DBLP:journals/pacmpl/BrotzmanZKT21}. The lack of
formal foundations for these liftings appears as a missed opportunity.
After all, some of the most successful approaches to reasoning about
constant-time and information flow are arguably three methodologies:
self-composition safety~\cite{DBLP:journals/mscs/BartheDR11},
product programs~\cite{DBLP:conf/fm/ZaksP08,Barthe2016,DBLP:conf/uss/AlmeidaBBDE16},
and dynamic taint analysis~\cite{ctgrind}.

This raises the questions of whether there exists a unifying method
for lifting CT verification to SCT verification, and whether the
method can be used to deliver new or better methods. This paper
answers both questions positively, by introducing a novel program
transformation. Our transformation, which we call speculation-passing
style (SPS) and denote \Tprogname{}, is defined such that a
program~\(c\) is SCT if and only if~\Tprog{c} is CT\@.

Speculation-Passing Style has two main benefits. First, cryptographic
implementers can continue using their favorite CT tools also for SCT
analysis, simply by applying the SPS transformation before analysis.
This also enables them to use deductive verification tools to prove SCT
for more sophisticated countermeasures, which was not possible before.
Second, tool developers can use our method as a blueprint to build SCT
tools and focus on CT verification, saving development and maintenance
effort. In this paper, we focus on establishing the theoretical
foundations of our approach and on showcasing its viability on
representative examples.

\paragraph{Detailed Contributions}
In this paper, we mostly focus on Spectre-v1 attacks, where the
attacker hijacks the branch predictor to take partial control over the
victim program's control flow, and the standard timing leakage model,
where conditional branches and memory accesses leak their guards and
addresses, respectively. Writing SCT code typically involves introducing
mitigations. Even though our approach can verify other
mitigations against Spectre-v1 attacks, e.g., index masking, we center
much of our discussion around \emph{selective speculative load
hardening} (SelSLH)~\cite{DBLP:conf/sp/ShivakumarBBCCGOSSY23}, a
countermeasure which refines LLVM's proposal for speculative load
hardening~\cite{carruth2018slh} and has been used to protect
high-assurance cryptography.

In this setting, we formalize the SPS program transformation and
analyze its theoretical foundations. We further substantiate this main
contribution with three additional contributions: two extensions of
the transformation (to fine-grained leakage models and to a different
Spectre variant); end-to-end methods for verifying SCT by combining
SPS with existing techniques; and an evaluation demonstrating their
feasibility. We structure our contributions, and this paper, as
follows.

\paragraph{Speculation-Passing Style}
The first, and main contribution of this paper, is the
speculation-passing style transformation, which is a sound and
complete method for reducing verification of SCT to verification
of CT\@.
At a high level, SPS internalizes speculative execution into
sequential execution, reminiscent of how continuation-passing style
internalizes returns as callbacks.
Remarkably, this reduction can be established without
loss of precision---this follows from the SCT threat model assuming that
the attacker completely controls the program's control flow.
More precisely, the reduction states that a program~\(c\) is speculative
constant-time if and only if \Tprog{c}~is constant-time.
To prove this result, we prove that the speculative leakage of a
program~\(c\) is in precise correspondence with the sequential leakage
of the program~\Tprog{c}.

\paragraph{Extensions}
We support the flexibility of our approach by discussing two extensions:
fine-grained leakage models and different Spectre variants.
First, we show how to combine our transformation~\Tprogname{} with
fine-grained leakage models~\cite{FinegrainedCT}, which include
the time-variable model (time-variable instructions leak their latency)
and the cache line model (memory accesses leak their cache lines rather
than the addresses).
Second, we discuss how to adapt our transformation to cover Spectre-v4
(store-to-load forwarding) attacks, which are notoriously difficult to
mitigate.

\paragraph{End-to-End SCT Verification Methods}
We provide end-to-end verification methods for (dis)proving SCT\@. These
methods combine SPS with existing verification techniques for CT\@.
First, we use a transformation that reduces CT to noninterference, a
property that can be verified with techniques such as Relational Hoare
Logic (RHL)~\cite{DBLP:conf/popl/Benton04}. Second, we use a
transformation that reduces CT to assertion safety, using the product
program construction~\cite{Barthe2016}. Third, we consider the
combination of SPS with dynamic analysis tools~\cite{ctgrind}.

\paragraph{Evaluation}
The last contribution of the paper is an evaluation of our approach
using the \easycrypt{} proof assistant~\cite{EasyCrypt}, the
\binsecrel{}~\cite{DBLP:conf/sp/DanielBR20} relational symbolic
execution tool%
\footnote{There exists an extension of \binsecrel{}, called
\binsechaunted{}, that aims to verify SCT directly. Our approach
provides an alternative that uses \binsecrel{}.  We compare them in
\cref{sec:evaluation}.}, and the \ctgrind{}~\cite{ctgrind} dynamic
taint analyzer.  We use these tools to analyze the SPS transform of
Kocher's benchmarks for SCT~\cite{benchmark}, an illustrating example for Spectre-v4, as well as two examples
from the literature: an SCT program that cannot be verified with an
SCT type system (from~\cite{ShivakumarTyping2023}), and a version of
the MAC rotation function in MEE-CBC encryption scheme of TLS 1.2 that
is secure in a weaker leakage model, called cache line leakage model,
and whose verification is intricate~\cite{FinegrainedCT}.  Our results
show that combining SPS with off-the-shelf CT verification tools is
feasible and enables the verification of examples that are beyond the
scope of existing SCT tools.

\paragraph{Organization}
\Cref{sec:overview} gives an overview of our approach, introducing CT
and SCT more precisely and illustrating our workflow with a simple
example beyond the capabilities of existing SCT tools.
\Cref{sec:setting} defines our source and target languages and security
notions formally.
\Cref{sec:transformation} defines the SPS transformation.
\Cref{sec:finegrained,sec:spectrev4} present two extensions of our
approach, to fine-grained leakage models and to Spectre-v4.
\Cref{sec:application} shows how SPS can be combined with three existing
approaches: reducing CT verification to noninterference verification,
reducing it to assertion safety verification, and applying dynamic
analysis.
\Cref{sec:evaluation} evaluates our approach using \easycrypt{},
\binsecrel{}, and \ctgrind{}.

\paragraph{Artifact}
We provide an artifact in
\ifanon{%
  \url{https://doi.org/10.5281/zenodo.17304903}%
}{%
  \url{https://doi.org/10.5281/zenodo.17339112}%
}.
The artifact contains original programs, SPS-transformed programs,
scripts, and mechanized proofs of our motivating example (\cref{sec:overview})
and evaluation (\cref{sec:evaluation}).

\linkappendix{}


\section{Overview}%
\label{sec:overview}
The definition of constant-time idealizes timing side-channel attacks by
means of an abstract leakage model, in which an attacker observes the
control flow and memory access addresses.
The CT policy requires that the sequences of observations generated by a
program's execution, called leakage traces, do not depend on secrets.
The SCT policy extends CT to the speculative setting: it requires that
leakage traces generated by the program's \emph{speculative execution}
are independent of secrets.
Our model of speculative execution conservatively assumes that the
attacker has complete control over the program's control flow, i.e.,
decides at every conditional or loop command which branch to take.
Attacker decisions are collected into a list of booleans, called
\emph{directives}~\cite{highassuranceSpectre}, that drive the program's
execution.

\begin{figure}
  \begin{subfigure}{0.4\linewidth}
    \centering
    \meminitinsecure{}
    \setcounter{subfigure}{0}
    \subcaption{Insecure Program.}%
    \label{fig:init:original}
  \end{subfigure}%
  \hfill%
  \begin{subfigure}{0.58\linewidth}
    \centering
    \meminitprotected{}
    \setcounter{subfigure}{2}
    \subcaption{SPS Transformation of the Insecure Program.}%
    \label{fig:init:sps-insecure}
  \end{subfigure}

  \bigskip

  \begin{subfigure}{0.4\linewidth}
    \centering
    \meminitspsinsecure{}
    \setcounter{subfigure}{1}
    \subcaption{Protected Program.}%
    \label{fig:init:protected}
  \end{subfigure}%
  \hfill%
  \begin{subfigure}{0.58\linewidth}
    \centering
    \meminitspsprotected{}
    \setcounter{subfigure}{3}
    \subcaption{SPS Transformation of the Protected Program.}%
    \label{fig:init:sps-protected}
  \end{subfigure}
  \caption{%
    The program in (\subref{fig:init:original}) is CT but vulnerable to
    Spectre-v1, and (\subref{fig:init:protected}) is its protected
    version using SelSLH\@.
    The programs in (\subref{fig:init:sps-insecure}) and
    (\subref{fig:init:sps-protected}) show the application of SPS to the
    previous two, which internalizes speculation as part of the
    program.%
  }%
  \label{fig:init}
\end{figure}

To illustrate the difference between CT and SCT,
\cref{fig:init:original} presents a program that writes a secret to an
array pointed by~\jcode{a} (\cref{line:example:store}) before some
computation, overwrites it with a public value~\jcode{pub}
(\cref{line:example:begin-while} to \cref{line:example:end-while}),
loads a value to~\jcode{v} (\cref{line:example:load}), and finally leaks
it (\cref{line:example:leak}).
This program is CT, since its branch conditions and memory accesses are
independent of the secret~\jcode{sec}.
However, it is vulnerable to Spectre-v1 attacks, which can poison the
branch predictor to speculatively skip the loop---i.e., to predict
that the branch condition~\jcode{i < 10} will evaluate to false in the
first iteration. Bypassing the initialization, the processor loads
the secret into~\jcode{v} in \cref{line:example:load} and then leaks
it in \cref{line:example:leak}.  Thus, the attacker can recover the
value of~\jcode{sec} because the effect of \cref{line:example:leak} on
the microarchitectural state persists even after the processor
realizes the misprediction and rolls back execution.

A possible mitigation against Spectre-v1 attacks is to insert a speculation
fence after every branch---virtually stopping speculation---at the cost of a
significant performance overhead.
This overhead can be reduced by introducing a \emph{minimal} number of
fences~\cite{blade}, but \cref{fig:init:protected} presents a more efficient
approach: \emph{selective speculative load hardening}.
The essence of selective speculative load hardening is to transform control flow
dependencies (which may be abused speculatively) into data-flow dependencies
(i.e., arithmetic operations, which are not affected by speculation).
Selective speculative load hardening uses a distinguished program variable, the
\emph{misspeculation flag} (MSF), to track whether execution is misspeculating
and mask values that the attacker may observe.
Since the MSF is only used by SelSLH operators, we leave it implicit in our
language.
In our examples, selective speculative load hardening is achieved using the
following three operations:
\begin{enumerate}[topsep=1ex,parsep=0pt,itemsep=0pt]
\item \jcode{\jasminprimitive{init\_msf}()} in \cref{line:slh-example:init-msf}:
  set the MSF to false (\msfnomask{}) and introduce a speculation fence;
\item \jcode{\jasminprimitive{update\_msf}(e)} in
  \cref{line:slh-example:update-msf}:
  set the MSF to true (\msfmask{}) if \(e\)~is false, else leave
  unchanged; and
\item \jcode{x = \jasminprimitive{protect}(e)} in
  \cref{line:slh-example:protect}:
  mask the expression~\(e\) w.r.t.\ the MSF (i.e., \jcode{x}~becomes~\(e\) if
  the MSF is~\msfnomask{} and otherwise a default value \jcode{0}).
\end{enumerate}
Thus, the attack on \cref{fig:init:original} is no longer applicable: exiting
the loop prematurely causes the \jcode{\jasminprimitive{protect}} in
\cref{line:slh-example:protect} to overwrite the value of \jcode{v} with a
default value, thereby revealing no useful information in
\cref{line:slh-example:leak}.
Consequently, the program in \cref{fig:init:protected} is speculative
constant-time: its control flow and memory accesses are independent of secrets.

Now, we want to prove that the protected program achieves SCT, and we
hope to leverage one of the several existing CT approaches from the
literature.  Our approach involves transforming the program so that we
can reason about it using standard---that is, nonspeculative--CT
verification techniques.  The crux of our transformation is to
consider the directives as additional inputs to the program, on which
we make no assumptions.  Specifically, we introduce a new input
variable~\dirvar{} that contains the list of directives.

\Cref{fig:init:sps-insecure} shows our speculation-passing style
transformation applied to the original program in
\cref{fig:init:original}.
The transformed program begins by assigning~\(\bot\) to a new
variable~\msvar{} in \cref{line:T-example:ms}, which reflects whether
execution is misspeculating.
Most commands are left unchanged, such as the memory store in
\cref{line:T-example:store}.
Each conditional branch in the original program is transformed into a
\jcode{\leakname} command first (leaking the condition of the branch, in
\cref{line:T-example:leak}) and then a modified conditional branch that
follows a directive instead (the new loop condition,
{\setlength{\fboxsep}{1pt}\colorbox{jasminhighlight}{highlighted}} in
\cref{line:T-example:loop}, reads from the \dirvar{}~input list).
Here, we write \jcode{\hd{\dirvar{}}} for the first element of the list.
\Cref{line:T-example:tl-ms1} discards the first element of~\dirvar{} and
updates~\msvar{} to track whether the prediction was incorrect---we
write \jcode{\tl{\dirvar{}}} for the tail of the list.
We must update \dirvar{}~and~\msvar{} similarly after exiting the loop
in \cref{line:T-example:tl-ms2}.

\Cref{fig:init:sps-protected} presents the transformation of the
protected program in \cref{fig:init:protected}, which contains SelSLH
operators.
\Cref{line:T-slh-example:init-msf} immediately stops execution if it is
misspeculating, corresponding to the fence behavior of
\jcode{\jasminprimitive{init\_msf}()}, and sets~\msf{} to~\msfnomask{}.
\Cref{line:T-slh-example:update-msf} updates the MSF with respect to the
argument of \jcode{\jasminprimitive{update\_msf}}, and
\cref{line:T-slh-example:protect} evinces the masking behavior of
\jcode{\jasminprimitive{protect}}.

As a result of our transformation, the \emph{sequential} leakage of
the SPS-transformed programs in
\cref{fig:init:sps-insecure,fig:init:sps-protected} precisely captures
the \emph{speculative} leakage of the original programs.
The program in \cref{fig:init:sps-insecure} is insecure (i.e., not CT)
when the first input directive in~\dirvar{} is~\(\bot\), corresponding
to the attack above.  On the other hand,
\cref{fig:init:sps-protected} is CT, since whenever~\dirvar{} causes a
misspeculation, the conditional assignment in
\cref{line:T-slh-example:protect} overwrites~\jcode{v}, matching the
functionality of SelSLH\@.

More generally, our approach builds on the following result, which we
revisit later in the paper.
We write \sbigsem{c(i)}{\seq{d}}{\seq{o}} for a speculative execution of
the program~\(c\) on input~\(i\) under directives~\seq{d} that generates
leakage~\seq{o}, and \tbigsem{c(i, \seq{d})}{\seq{o}} for a sequential
execution of the program~\(c\) with inputs~\(i, \seq{d}\) that generates
observations~\seq{o}.

\begin{theorem*}
  There exists a leakage transformation function~\Tprogobsname{} such
  that
  \begin{equation*}
    \sbigsem{c(i)}{\seq{d}}{\seq{o}}
    \iff
    \tbigsem{\Tprog{c}(i, \seq{d})}{\Tprogobs{\seq{o}}{\seq{d}}}\text.
  \end{equation*}
\end{theorem*}
This theorem entails the desired property: a program~\(c\) is SCT if
and only if its SPS transform \Tprog{c} is CT\@. The key intuition is as
follows.
A program~\(c\) is SCT (w.r.t.\ a given relation on inputs) if for every
related inputs \(i_1\)~and~\(i_2\), and for every list of directives
\seq{d}, speculative execution of~\(c\) under directives~\seq{d}
yields equal leakage with inputs \(i_1\)~and~\(i_2\). Similarly, a
program~\(c\) is CT (w.r.t.\ a given relation on inputs) if for every
related inputs \(i_1\)~and~\(i_2\), sequential execution of~\(c\)
yields equal leakage with inputs \(i_1\)~and~\(i_2\). Therefore, it
suffices to prove that for every \(d\), \(o_1\), and \(o_2\), we have
$\Tprogobs{\seq{o_1}}{\seq{d}} = \Tprogobs{\seq{o_2}}{\seq{d}}$ if
and only if $\seq{o_1}=\seq{o_2}$, which can be established by
inspecting the definition of \Tprogobsname{} in the proof of the
theorem.

\paragraph{Combination with Existing Techniques}
An important benefit of the SPS transformation is its compatibility with
existing CT verification techniques.
We consider three such combinations: verification via noninterference
(\cref{sec:application:noninterference}), verification via assertion safety
(\cref{sec:application:assertion-safety}) and finally verification via
dynamic analysis (\cref{sec:application:taint-analysis}).

\begin{figure}
  \centering
  \meminitTprotected{}
  \caption{%
    Transformation of the program in \cref{fig:init:sps-protected} for
    noninterference verification.%
  }%
  \label{fig:verif}
\end{figure}

\Cref{fig:verif} presents the program in
\cref{fig:init:sps-protected} after standard leakage instrumentation and
assert elimination transformations, resulting in a simple imperative
program supported by techniques such as Relational Hoare Logic.
In a nutshell, the transformations introduce two ghost variables:
\retvar{}, which tracks whether the program should return, and
\obsvar{}, which accumulates the leakage of the program.
Assertions are replaced by conditional assignments to~\retvar{}, and the
rest of the code is guarded by~\retvar{} to skip execution after a
failed assertion. The leakage instrumentation appends to~\obsvar{} the
information leaked after every command.
Thanks to these transformations, our artifact uses an implementation of
RHL to verify that the program in \cref{fig:init:protected} is SCT by
establishing noninterference of \cref{fig:verif}.
(See \path{initialization.ec} in the artifact.)

On the other hand, we can also combine SPS with techniques that reduce
CT to assertion safety and dynamic taint analysis.
For assertion safety, we use the product program construction~\cite{Barthe2016} to obtain a
program that is assertion safe if and only if the original program (in
\cref{fig:init:original}) is SCT\@.
For dynamic taint analysis, we randomly generate inputs and directives
to dynamically check whether secrets will be leaked at branches or memory accesses.
Our artifact realizes these two methodologies with a symbolic execution tool and a dynamic taint analysis tool that
automatically find the vulnerability in \cref{fig:init:original} by
analyzing \cref{fig:init:sps-insecure}.
(See \path{initialization.c} in the artifact.)


\section{Language and Security}%
\label{sec:setting}
\Cref{sec:setting:source} introduces a core imperative language with speculative
semantics, and \cref{sec:setting:security} formalizes the security notion of
speculative constant-time.
Afterward, \cref{sec:setting:target} introduces the target language of our
transformation: a minimal imperative language with standard sequential
semantics.
Throughout the paper, we make the standard assumption that programs are type
safe, that is, that expressions always evaluate to a value of the expected type.

\subsection{Source Language}%
\label{sec:setting:source}
The syntax of our source language comprises expressions and commands, defined
as follows:
\begin{align*}
  e &\Coloneqq n \mid b \mid x \mid {\oplus}(e, \cdots, e)\\
  c & \Coloneqq
    \iassign{x}{e} \mid \iload{x}{e} \mid \istore{e}{x}
    \mid \iinitmsf{} \mid \iupdatemsf{e} \mid \iprotect{x}{e}\\
    &\quad \mid \iif{e}{c}{c} \mid \iwhile{e}{c} \mid \cskip \mid c \capp c
\end{align*}
where
\(n\)~is a natural number,
\(b\)~is a boolean,
\(x\)~is a variable, and
\(\oplus\)~is an operator such as \(+\)~or~\(\land\).
We assume that operators are deterministic and have no side effects---note that
there are no memory accesses in the expressions.
The values of this language are integers and booleans.
We write \iift{e}{c} for \iif{e}{c}{\cnil}.

States are quadruples \st{c}{\vm}{\mem}{\ms} consisting of a command~\(c\),
a \emph{variable map}~\vm{} (a function from variables to values), a
memory~\mem{} (a function from natural numbers to values), and a
\emph{misspeculation status}~\ms{} (a boolean tracking whether misspeculation
has happened).
We call a state a \emph{misspeculating state} when \ms{} is~\msmiss{}.
Executions start from \emph{inputs}~\emath{i = (\vm, \mem)}, which are
pairs of variable maps and memories: we write \emath{c(i)} for the initial state
of program~\(c\) on input~\(i\), defined as \st{c}{\vm}{\mem}{\msseq}.

We capture the standard CT leakage model by indexing our semantics with
\emph{observations}~\oset{} that correspond to the two operations that leak:
every conditional branch---i.e., if and while---produces the observation
\obranch{b}, where~\(b\) is the value of their condition; and
every memory access produces the observation~\oaddr{i}, where~\(i\) is the
address accessed.

To model the adversarial control of the branches, we index our semantics by
\emph{directives}~\dset{}~\cite{highassuranceSpectre}, which steer control
flow.
There are two directives, \dforce{\top} and \dforce{\bot}, which force the
execution of the then and else branches, respectively.

\begin{figure}
  \figspecsem{}
  \caption{Speculative semantics of the source language.}%
  \label{fig:source-semantics}
\end{figure}

\Cref{fig:source-semantics} presents the semantics of our language.
We write \sem{s}{\seq{d}}{\seq{o}}{s'} to indicate that the state~\(s\)
performs one step of execution under the directive list~\seq{d}, and
produces the observation list~\seq{o} and the resulting state~\(s'\).

The \refsemassign{} rule is standard: it assigns the
value of the right-hand side~\(e\) to the left-hand side~\(x\).
It consumes no directives and produces no observations.
The \refsemload{} rule states that if the instruction
under execution is a load, we evaluate the expression~\(e\) to get an
address~\(i\), and we write the value stored at that address in memory to the
variable~\(x\).
Loads consume no directives and leak their addresses with the
\oaddr{i}~observation---we write \emath{[\oaddr{i}]} to emphasize that it is a
list with one element.
The \refsemstore{} rule is analogous.
The \refsemcond{} rule illustrates the purpose of directives: it
consumes~\dforce{b}, which corresponds to a prediction from the branch
predictor, and follows that branch, regardless of the evaluation of its
condition.
The value of the condition leaks with the observation~\obranch{b'},
and \ms{} becomes \msmiss{} if the prediction was incorrect.
The \refsemwhile{} rule is similar.

The behavior of SelSLH operators is as described in \cref{sec:overview}.
The \refseminit{} rule models a speculation fence by requiring the
misspeculation status of the state on the left-hand side to be~\msseq{}.
It also sets the MSF~\msf{} to~\msfnomask{}.
The \refsemupdate{} rule updates the MSF according to the value of~\(e\), and
the \refsemprotect{} rule masks the value of~\(e\) according to the MSF and
assigns it to~\(x\).

As defined in \refsemrefl{} and \refsemtrans{}, we write
\sem*{s}{\seq{d}}{\seq{o}}{s'} for executions of zero or more steps.
We say that a state is \emph{final}, written \final{s}, if it is of the form
\st{\cnil}{\vm}{\mem}{\ms} or \st{c}{\vm}{\mem}{\top} where the first instruction of $c$ is $\iinitmsf$.
Finally, \refsembig{} states that we write \sbigsem{s}{\seq{d}}{\seq{o}} for
complete executions starting from~\(s\) under directives~\seq{d} producing
observations~\seq{o}.

\subsection{Speculative Constant-Time}%
\label{sec:setting:security}
We can now define SCT precisely: a program is speculative constant-time if its
observations under speculative execution are independent of secrets.
In line with the standard definition of noninterference, we require that if two
inputs are indistinguishable they produce the same observations.

\begin{definition}[\phiSCT{}]%
\label{def:sct}
  A program~\(c\) is \emph{speculative constant-time} w.r.t.\ a relation~\phi{}
  (denoted \phiSCT{}) if it produces the same observations for every list of
  directives and every pair of related inputs.
  That is, for every \(i_1\), \(i_2\), \seq{d}, \seq{o_1}, and \seq{o_2}, we
  have that
  \begin{equation*}
    i_1 \relphi i_2 \land
    \sbigsem{c(i_1)}{\seq{d}}{\seq{o_1}} \land
    \sbigsem{c(i_2)}{\seq{d}}{\seq{o_2}} \implies
    \seq{o_1} = \seq{o_2}\text.
  \end{equation*}
\end{definition}

The relation on inputs~\phi{} encodes the \emph{low-equivalence} of the inputs,
i.e., that \(i_1\)~and~\(i_2\) coincide in their public part.
For example, for the memory initialization program in \cref{sec:overview},
we should define \phi{} as
\begin{equation*}
  (\rho_1, \mu_1) \relphi (\rho_2, \mu_2) \eqdef
    \rho_1(\mathsf{n}) = \rho_2(\mathsf{n})\text,
\end{equation*}
meaning that~\jcode{n} is public and all other variables (in particular,
\jcode{sec}) are secret.

\subsection{Target Language}%
\label{sec:setting:target}
The syntax of the target language is that of the source \emph{without}
the SelSLH operators and with a new instruction: \iassert{e}.
This command steps to an error state~\sterror{} when its condition~\(e\)
is false.

Target states are triples \st{c_t}{\vm_t}{\mem_t}{} consisting of
a target command~\(c_t\), a variable map~\(\vm_t\), and a memory~\(\mem_t\).
Target variable maps associate variables to values or lists of values.
Executions in this language take a list of directives as an extra input.
Thus, given an input~\emath{i = (\vm, \mem)} and a list of directives \seq{d},
the \emph{initial state} of a program~\(c\) is
\emath{c(i, \seq{d}) \eqdef \st{c}{\mset{\vm}{\dirvar}{\seq{d}}}{\mem}{}},
where~\dirvar{} is a distinguished variable that does not occur in source
programs.

\refappendixtgtsem{} presents the
semantics of the target language. The form
\sem{t}{}{\seq{o}}{t'} is for one step of execution from a state~\(t\) to a
state~\(t'\) producing observations~\seq{o}.
And, $\sem*{t}{}{\seq{o}}{t'}$ is for zero or more steps of execution, whereas
$\nsem{t}{}{\seq{o}}{t'}{n}$ is for exactly $n$ steps of execution.
We write \tbigsem{t}{\seq{o}} for
\qE{t'}{\sem*{t}{}{\seq{o}}{t'} \land \final{t'}}, where a target state
is \emph{final} if its code is~\cskip{} and the value of the~\dirvar{} variable
is the empty list or it is an error state.

\subsection{Constant-Time}%
\label{sec:setting:ct}
Finally, we can define our security notion for the target language.

\begin{definition}[\phiCT{}]%
\label{def:ct}
  A program~\(c\) is \emph{constant-time} w.r.t.\ a relation on
  inputs~\phi{} (denoted \phiCT{}) if it produces the same observations
  for every pair of related inputs.
  That is, for every \(i_1\), \(i_2\), \seq{d}, \seq{o_1}, and
  \seq{o_2}, we have that
  \begin{equation*}
    i_1 \relphi i_2 \land
    \tbigsem{c(i_1, \seq{d})}{\seq{o_1}} \land
    \tbigsem{c(i_2, \seq{d})}{\seq{o_2}}
    \implies
    \seq{o_1} = \seq{o_2}
    \text.
  \end{equation*}
\end{definition}

Note that the executions in these statements may terminate in an error
state.


\section{Speculation-Passing Style}%
\label{sec:transformation}
This section presents our transformation, Speculation-Passing Style (SPS), which
materializes the speculative behavior of a program in a sequential language.
In order for our transformation to capture speculative execution, we assume
two distinguished program variables in target states:
\msvar{}, which captures the behavior of~\ms{}, and~\dirvar{}, which
stores the remaining directives.

\begin{figure}
  \figtrans{}
  \caption{Speculation-Passing Style transformation.}%
  \label{fig:trans}
\end{figure}

\Cref{fig:trans} present our SPS transformation, denoted~\Tprogname{},
which transforms a program in the source syntax presented in
\cref{sec:setting:source} into one in the target syntax presented in
\cref{sec:setting:target}.  We first define an auxiliary
transformation, \Tprogstarname{}, inductively on the code.  The basic
commands (i.e., assignment, load, store, and \cskip{}) are left
unmodified.  Conditionals first leak their branch condition with
\ileak{e}, which is notation for \iif{e}{\cskip}{\cskip}.  Then, we
extract a directive---i.e., an adversarially controlled branch
prediction---from the list~\dirvar{} with \hd{\dirvar} and follow that
branch.  Inside each branch, we pop an element of~\dirvar{} with
\tl{\dirvar} and update~\msvar{} to~\msmiss{} if the prediction was
incorrect.  Loops are transformed similarly.  The MSF initialization
command is transformed as an assertion that execution is not
misspeculating---modeling its fence behavior---followed by setting the
MSF variable~\msf{} to~\msfnomask{}.  The MSF update command sets the
\msf{}~variable to \msfmask{} if its argument~\(e\) is false.  The
last SelSLH operator, protect, masks its argument~\(e\) according to
the value of~\msf{}.  Lastly, the transformation of sequencing is as
expected.  Finally, the transformation~\Tprogname{} for whole programs
initializes~\msvar{} to~\msseq{} before transforming the body of the
program with~\Tprogstarname{}.
Note that the transformation replaces each statement with a fixed number
of statements; thus the code size increases linearly, i.e., there is no
code blowup.

Let us now turn to the theorem that underpins our approach: a transformed
program~\Tprog{c} sequentially matches the speculative behavior of~\(c\).

\begin{theorem}[Soundness and Completeness of SPS]%
\label{thm:correct-trans}
  There exists a function~\emath{
    \Tprogobsname : \observations^* \times \directives^* \to \observations^*
  }, which is injective in its first argument, such that for any
  program~\(c\), input~\(i\), directive sequence~\seq{d}, and
  observation sequence~\seq{o}, we have that
  \begin{equation*}
    \sbigsem{c(i)}{\seq{d}}{\seq{o}}
    \iff
    \tbigsem{\Tprog{c}(i, \seq{d})}{\Tprogobs{\seq{o}}{\seq{d}}}\text.
  \end{equation*}
\end{theorem}
\begin{proof}
  Follows immediately from~\refappendixspsthm{}, which proves a stronger statement that also relates
  the final states.
\end{proof}

The definition of~\Tprogobsname{} in the above theorem follows from the
observation that since every conditional \iif{e}{c}{c'} in the source program is
transformed into~\emath{\ileak{e}\capp \iif{\hd{\dirvar}}{{\dots}}{{\dots}}},
every source observation~\obranch{b} induces two observations in the target
program, namely \obranch{b} and \obranch{(\hd{\dirvar})}.
The situation is analogous for loops.
Consequently, we define~\Tprogobsname{} as follows:
\begin{equation*}
  \Tprogobs{\seq{o}}{\seq{d}} \eqdef
  \begin{cases}
    \lnil & \when{\seq{o} = \lnil}\text,\\
    \oaddr{n} \lapp \Tprogobs{\seq{o}'}{\seq{d}}
    & \when{\seq{o} = \oaddr{n} \lapp \seq{o}'}\text,\\
    \obranch{b} \lapp \obranch{b'} \lapp \Tprogobs{\seq{o}'}{\seq{d}'}
    & \when{
        \seq{o} = \obranch{b} \lapp \seq{o}' \text{ and }
        \seq{d} = \dforce{b'} \lapp \seq{d}'\text.
      }
  \end{cases}
\end{equation*}
That is, \Tprogobs{\seq{o}}{\seq{d}} inserts, for every index~\(i\), the
observation \obranch{b'} immediately after the \(i\)-th branch observation,
where the \(i\)-th directive is \dforce{b'}.

Now, we can reduce the verification of speculative constant-time to the
verification of constant-time, as follows.

\begin{corollary}[Reduction of SCT to CT]%
\label{thm:ct-to-sct}
  A program~\(c\) is \phiSCT{} if and only if \Tprog{c} is \phiCT{}.
\end{corollary}
\begin{proof}
  The backward implication follows immediately from \cref{thm:correct-trans}:
  given \sbigsem{c(i_k)}{\seq{d}}{\seq{o_k}} for \emath{k \in \set{1, 2}}, we
  have \tbigsem{\Tprog{c}(i_k, \seq{d})}{\Tprogobs{\seq{o_k}}{\seq{d}}}, and
  the \phiCT{} hypothesis gives us
  \emath{\Tprogobs{\seq{o_1}}{\seq{d}} = \Tprogobs{\seq{o_2}}{\seq{d}}}, which
  means that \emath{\seq{o_1} = \seq{o_2}} by injectivity.

  The forward implication entails showing that \Tprogobsname{} is surjective for
  executions of transformed programs, i.e., that
  \tbigsem{\Tprog{c}(i, \seq{d})}{\seq{o}} implies that there exists
  \emath{\seq{o}'} such that \emath{\Tprogobs{\seq{o}'}{\seq{d}} = \seq{o}}.
  This surjectivity follows from the structure of~\Tprog{c}.
  Consequently, given \tbigsem{c(i_k, \seq{d})}{\seq{o_k}} for
  \emath{k \in \set{1, 2}}, we have
  \tbigsem{c(i_k, \seq{d})}{\Tprogobs{\seq{o_k}'}{\seq{d}}} by surjectivity,
  and \cref{thm:correct-trans} gives us the source executions
  \sbigsem{\Tprog{c}(i_k)}{\seq{d}}{\seq{o_k}'}.
  Finally, the \phiSCT{} hypothesis shows \emath{\seq{o_1}' = \seq{o_2}'},
  which means that \emath{
    \seq{o_1} =
    \Tprogobs{\seq{o_1}'}{\seq{d}} =
    \Tprogobs{\seq{o_2}'}{\seq{d}} =
    \seq{o_2}
  }.
\end{proof}


\section{Fine-Grained Leakage Models}%
\label{sec:finegrained}
This section discusses how to extend our approach to other leakage
models beyond the \emph{baseline} constant-time leakage model that we
have considered so far. Concretely, the baseline leakage model has
been generalized both in the literature and in practice to account for
more specific threat models.  These generalizations are instances of
\emph{fine-grained} constant-time leakage models~\cite{FinegrainedCT}.
Below, we discuss two examples and CT verification in this setting.

The \emph{variable-time} leakage model assumes that operators leak
information about their operands, thus preventing attacks that exploit
variable-time instructions.
For example, the execution time of a division operation depends on the
value of its operands in many CPUs, a fact that has been recently
exploited in Kyberslash~\cite{kyberslash}.
The variable-time leakage model is stricter than the baseline model, as
it considers that an assignment leaks a function of the values of its
expression; for example, \iassign{x}{a / b} leaks the sizes of
\(a\)~and~\(b\) instead of~\lnil{}.

Another example is the \emph{cache-line} leakage model, which assumes
that the attacker can observe the cache line of an address being
accessed but not the address itself.
Specifically, accessing an address~\(a\) leaks \floor{\frac{a}{N}},
where \(N\)~is the size of a cache line, e.g., 64 bytes.
This model is more realistic and more permissive than the baseline
model, i.e., there are secure programs in the cache-line model that are
deemed insecure in the baseline model.
The advantage of permissive models is that they allow for more
optimizations: some cryptographic libraries, such as OpenSSL, offer
multiple implementations of the same function, optimized
w.r.t.\ different leakage models, allowing users to choose their
trade-off between efficiency and security.

\begin{figure}
  \small
  \begin{mathpar}
    \inferrule[\labelsemassign]{
      v = \eval{e}{\vm}\\
      v' = \eval{e'}{\vm}\\
      \vm' = \mset{\vm}{x}{\floor{\frac{v}{v'}}}
    }{
      \sem{
        \st{\iassign{x}{\frac{e}{e'}}}{\vm}{\mem}{\ms}
      }{\dstep}{\log_2(v),\, \log_2(v')}{
        \st{\cskip}{\vm'}{\mem}{\ms}
      }
    }

    \inferrule[\labelsemload]{
      i = \eval{e}{\vm}\\
      \vm' = \mset{\vm}{x}{\mem(i)}
    }{
      \sem{
        \st{\iload{x}{e}}{\vm}{\mem}{\ms}
      }{\dstep}{[\oaddr{\floor{\frac{i}{64}}}]}{
        \st{\cskip}{\vm'}{\mem}{\ms}
      }
    }

    \inferrule[\labelsemstore]{
      i = \eval{e}{\vm}
    }{
      \sem{
        \st{\istore{e}{x}}{\vm}{\mem}{\ms}
      }{\dstep}{[\oaddr{\floor{\frac{i}{64}}}]}{
        \st{\cskip}{\vm}{\mset{\mem}{i}{\vm(x)}}{\ms}
      }
    }
  \end{mathpar}
  \caption{%
    Selected rules of the fine-grained semantics of the source
    language.%
  }%
  \label{fig:sem:finegrained}
\end{figure}

Fine-grained leakage models, such as the two above, are more challenging
for verification than the baseline model, since they require reasoning
about values rather than mere dependencies.
Consequently, these models often require deductive methods of CT
verification.

\paragraph{Semantics for Fine-Grained Leakage Models}
Fine-grained leakage models are formalized in terms of two functions
that define the leakage of operators and memory accesses, denoted
\Lopname{}~and~\Laddrname{}, respectively.
These functions allow us to generalize the semantics from
\cref{sec:setting:source,sec:setting:target}: assignments \iassign{x}{e}
leak the value of~\Lop{e} (instead of producing no observation), and
load instructions \iload{x}{e} (or stores \istore{e}{x}) leak the value
of~\Laddr{e} (instead of the value of~\(e\)).
Generalizing the semantics in this way gives rise to a refined notion of
speculative constant-time, which we call \phiSCT{}
w.r.t.\ \Lopname{}~and~\Laddrname{}, and analogously with \phiCT{}.

\Cref{fig:sem:finegrained} illustrates a fine-grained leakage semantics.
The \refrule{fig:sem:finegrained}{\labelsemassign} rule corresponds to the
variable-time leakage model, where the assignment divides \(e\)~by~\(e'\)
and, therefore, leaks the sizes (i.e., the number of bits) of both operands.
Thus, the function \Lopname{} maps division expressions to sizes of their
operands.
On the other hand, the \refrule{fig:sem:finegrained}{\labelsemload} and
\refrule{fig:sem:finegrained}{\labelsemstore} rules correspond to the
cache-line leakage model, where both loads and stores leak the cache line they
access, assuming a size of 64 bytes.
Thus, the function \Laddrname{} maps addresses to their cache lines.

\paragraph{Extending SPS to Fine-Grained Leakage Models}
Fortunately, the only changes to the semantics are in the observations,
which means that the SPS transformation requires no changes to handle
fine-grained leakage models.
Thus, SPS satisfies the following reduction theorem for fine-grained
leakage models.

\begin{lemma}[Reduction of Fine-Grained SCT to Fine-Grained CT]%
  \label{lem:fg-ct-to-sct}
  A program~\(c\) is \phiSCT{} w.r.t.\ \Lopname{}~and~\Laddrname{} if
  and only if \Tprog{c} is \phiCT{} w.r.t.\ \Lopname{}~and~\Laddrname{}.
\end{lemma}

This reduction enables the verification of fine-grained SCT using
deductive verification methods for fine-grained CT\@; as mentioned
before, such tools were not previously available for SCT.


\section{Speculation-Passing Style for Spectre-v4}%
\label{sec:spectrev4}

This section presents an extension of our approach to
Spectre-v4~\cite{spectre}.
This Spectre variant exploits the
store-to-load forwarding predictor in modern CPUs to recover secrets
even if they should have been overwritten.
Specifically, it forces the processor to mispredict that a load does not
depend on a preceding store to the same address.
Thus, the processor speculatively loads the stale value before the store 
is completed.

\Cref{fig:v4:insecure} presents a minimal example program that is
vulnerable to Spectre-v4.
This program stores a secret value~\jcode{sec} at address~\jcode{a}
in \cref{line:v4:store-sec}, then overwrites it with a public
value~\jcode{pub} in \cref{line:v4:store-pub}, loads a value from the
same address in \cref{line:v4:load}, and finally leaks the loaded value
in \cref{line:v4:leak}.
If the processor mistakenly predicts that the load in
\cref{line:v4:load} does not depend on the preceding public store, it
will execute the load without waiting the store to complete.
Hence, it will speculatively load the secret value~\jcode{sec} stored
earlier at address~\jcode{a} and leak it in \cref{line:v4:leak}.




\begin{figure}[]
  \centering
  \captionsetup{width=.7\linewidth}
  \begin{minipage}{0.7\linewidth}
    \centering
    \begin{subfigure}{.45\linewidth}
      \centering
      \examplevfour{}
      \caption{Vulnerable code.}%
      \label{fig:v4:insecure}
    \end{subfigure}%
    \hspace{0.05\linewidth}%
    \begin{subfigure}{.45\linewidth}
      \centering
      \exampleTvfour{}
      \caption{Protected code.}%
      \label{fig:v4:protected}
    \end{subfigure}%
  \end{minipage}

  \caption{%
    The program in (\subref{fig:v4:insecure}) is vulnerable to
    Spectre-v4, and (\subref{fig:v4:protected}) is its protected
    version using a speculation fence.%
  }%
  \label{fig:v4}
\end{figure}

This attack can be thwarted with a speculation fence before the leaking
instruction, as shown in \cref{fig:v4:protected}.
The \initmsfname{} in \cref{line:Tv4:fence} prevents subsequent
instructions from being affected by the (mispredicted) reordering of the
load with the public store.
Using fences as a mitigation incurs a considerable performance penalty;
however, it is standard practice in cases where it is impossible to
disable store-to-load forwarding (e.g., by setting the \texttt{SSBD}
flag~\cite{ssbd}).

\paragraph{Semantics for Spectre-v4}
We now extend our source semantics from \cref{sec:setting:source} to
account for Spectre-v4, following the style of the semantics
in~\cite{highassuranceSpectre}.
Intuitively, the extension considers that reading from memory may return
one of the many values previously stored at that address.
Concretely, memories in this semantics map each address to a
\emph{list}, i.e., \emath{\mem : \N \to \N^*}, containing all the values
stored at that address.
We extend directives with a \dload{n}~directive, which forces a load
instruction to load the \(n\)-th most recent value stored at its
address, allowing the attacker to forward from any earlier store.

\begin{figure}
  \figssbsem{}
  \caption{Speculative semantics for Spectre-v4.}%
  \label{\labelfigssbsem}
\end{figure}

\Cref{\labelfigssbsem} presents the three modified rules that extend our
source semantics to Spectre-v4.
The \refssbsemload{} rule evaluates its address to~\(i\) as before, but
uses the \dload{n}~directive to determine which value to load from that
address.
Thus, the semantics loads \(n\)-th element of \emath{\mem(i)}, denoted
\emath{\mem(i)_{n}}.
Additionally, the misspeculation status becomes \msmiss{} unless the
loaded value is the most recent store to the address.
The observation is the address loaded from, as before.
The \refssbsemstore{} rule appends the stored value to the list
corresponding to the address; it leaks the address as before.
Finally, the \refssbseminit{} rule requires that the misspeculation
status is \msseq{} and sets the MSF to \msfnomask{}, as before.
However, it now modifies the resulting memory to discard all but the
most recent stores to each address.

We extend \cref{def:sct} to redefine \phiSCT{} using the modified
semantics described above.

\paragraph{Speculation-Passing Style for Spectre-v4}


\begin{figure}
  \figssbtrans{}
  \caption{Speculative-Passing Style Transformation for Spectre-v4.}%
  \label{fig:trans-ssb}
\end{figure}

The SPS transformation for Spectre-v4 directly reflects the changes in
the semantics.
\Cref{\labelfigssbsem} presents the three cases that differ from the
Spectre-v1 version.
Load commands now read the value at position \jcode{\hd{\dirvar}} from
their memory location, and store commands append an element to theirs.
The case for \initmsfname{} is transformed as before, and also
introduces a new \clearmemname{}~command, which sets each memory
location to its most recent element.
Intuitively, \clearmemname{} is equivalent to
\emath{\istore{0}{[\hd{[0]}]}\capp \istore{1}{[\hd{[1]}]}\capp \dots},
resetting every location.

This transformation enjoys a soundness and correctness result similar
to \cref{thm:correct-trans} that relates the speculative leakage of a
program~\(c\) with the sequential leakage of~\Tssb{c}.
This correspondence is established by a function, analogous
to~\Tprogobsname{} in the case of Spectre-v1.
Consequently, we can prove a reduction theorem between SCT and CT that
covers Spectre-v4 as follows.

\begin{theorem}[Reduction of SCT to CT for Spectre-v4]%
\label{thm:ssb:ct-to-sct}
  A program~\(c\) is \phiSCT{} in the Spectre-v4 semantics if and only
  if \Tssb{c}~is \phiCT{}.
\end{theorem}


\section{End-to-End Speculative Constant-Time Methods}%
\label{sec:application}
In this section, we present different ways in which SPS can be
combined with other methods to (dis)prove CT\@.

\begin{theorem}[%
    Soundness and Completeness of Assertion Elimination for CT%
  ]%
\label{thm:assert-elim-ct}
  A program~\(c\) is \phiCT{} if and only if \Tass{c} is \phiCT{}, for
  any \phi{}.
\end{theorem}

This theorem follows directly from a result analogous to
\cref{thm:correct-trans}, where the (sequential) observations of~\(c\)
and~\Tass{c} are in precise correspondence.

\subsection{Verification via Noninterference}%
\label{sec:application:noninterference}
Our first approach is based on a---folklore---reduction of constant-time
to noninterference, a widely studied information flow policy which
requires that public outputs do not depend on secret inputs.

The reduction is performed in two steps.
The first step removes \iassertname{} commands, by making programs single-exit,
i.e., only return at the end of the program, with a standard
transformation, denoted \Tassname{} and presented in
\refappendixae{}, that introduces conditional branches to
skip the rest of the program when an assertion fails.
The second step introduces a ghost variable that accumulates leakage during
execution (see, e.g.,~\cite{FinegrainedCT}). \Cref{fig:leak-inst}
presents such a transformation, denoted~\Tleakname{}, which
initializes a ghost variable~\obsvar{} to the empty list and appends
to it the observation resulting from each memory access and branch.
All other instructions are left unchanged.

\begin{figure}
  \figleaktrans{}
  \caption{Leakage instrumentation.}%
  \label{fig:leak-inst}
\end{figure}

Both steps are sound and complete for CT\@. Consequently, we can employ
standard noninterference verification techniques to prove SCT\@.  In
this work, we choose Relational Hoare Logic (RHL), is sound and
complete to prove noninterference~\cite{Barthe2016}.

Recall that RHL manipulates judgments of the form \rhl{c_1}{c_2}{\Pre}{\Pos},
which mean that terminating executions of the programs \(c_1\)~and~\(c_2\)
starting from initial states in the relational precondition~\Pre{} yield final
states in the relational postcondition~\Pos{}.
Thus, the statement for CT soundness and completeness introduced in previous
work is as follows.

\begin{theorem}[Soundness and Completeness of RHL for CT]
  A program~\(c\) is \phiCT{} if and only if we can derive
  \begin{equation*}
    \rhl%
      {\Tleak{c}}
      {\Tleak{c}}
      {\phi}
      {\eqset{\obsvar}}\text,
  \end{equation*}
  where the \eqset{\obsvar} clause means that the final states coincide
  on the variable~\obsvar{}.
\end{theorem}
Combining this result with the soundness and completeness of SPS
(\cref{thm:ct-to-sct}), we obtain the following verification methodology.

\begin{corollary}[Soundness and Completeness of RHL for SCT]
  A program~\(c\) is \phiSCT{} if and only if we can derive
  \begin{equation*}
    \rhl%
      {\Tleak{\Tass{\Tprog{c}}}}
      {\Tleak{\Tass{\Tprog{c}}}}
      {\phi \land \eqset{\dirvar}}
      {\eqset{\obsvar}}\text,
  \end{equation*}
  where the \eqset{\dirvar} clause means that the initial states
  coincide on the variable~\dirvar{}.
\end{corollary}
A stronger result, using a restricted version of RHL with only two-sided
rules, also holds.

\subsection{Verification via Assertion Safety}%
\label{sec:application:assertion-safety}

\begin{figure}
  \figproducttrans{}
  \caption{Product program transformation.}%
  \label{fig:product-program}
\end{figure}

A different approach to verify CT is to reduce it to the assertion
safety of a product program (see, e.g.,~\cite{Barthe2016}). Again, the
reduction is performed in two steps. The first step is to replace
\iassertname{} commands by conditioned \ireturnname{} commands; by
abuse of notation, we also call this step $\Tass{}$. The second step
is to build the product program. Intuitively, the product program is
obtained by duplicating every instruction of the original program
(introducing two copies of each variable) and inserting an assertion
before every memory access and conditional branch (which guarantees
that the leaked values coincide).  \Cref{fig:product-program}
summarizes the transformation, denoted~\Tprodname{}, which depends on
the indistinguishability relation \phi{}.  We use \rtag{e}{n} to
rename all variables in the expression~\(e\) with subscript~\(n\).  In
this setting, the soundness and completeness statement for CT
from~\citet{Barthe2016} is as follows.

\begin{theorem}[Soundness and Completeness of Assertion Safety for CT]
  A program~\(c\) is \phiCT{} if and only if \Tprod{c}{\phi} is assertion safe,
  i.e.,
  \begin{equation*}
    \qnE{\,i \,\seq{d}}{\sem*{\Tprod{c}{\phi}(i, \seq{d})}{}{}{\sterror}}\text.
  \end{equation*}
\end{theorem}

The following corollary establishes that composing the product program
transformation with SPS lifts this methodology to SCT\@.

\begin{corollary}[Soundness and Completeness of Assertion Safety for SCT]
  A program~\(c\) is \phiSCT{} if and only if
  \Tprod{\Tass{\Tprog{c}}}{\phi \land \eqset{\dirvar}} is assertion safe.
\end{corollary}

\subsection{Verification by Dynamic and Taint Analysis}%
\label{sec:application:taint-analysis}
Dynamic analysis and fuzzing execute the program with different secret
inputs and collect observations~\cite{WichelmannMES18} (e.g., memory
access patterns or taken branches) or execution
times~\cite{DBLP:conf/date/ReparazBV17}, and aim to find differences
between them.
When such differences are found, the tools report a potential CT
violation.
Taint analysis (e.g.,~\cite{ctgrind}) is a kind of dynamic approach, which
marks secret inputs as ``tainted'' and tracks how the taint propagates
through the program.
If the taint reaches a memory access or a branch, the tool reports a CT
violation.
Thus, taint analysis has the benefit of not requiring multiple executions
with different secret inputs.
With the help of SPS, which exposes speculative execution as sequential execution, dynamic and taint analysis tools will be able to access speculative leakages.
In this way, both dynamic analysis and taint analysis can be lifted to
SCT verification.


\section{Evaluation}%
\label{sec:evaluation}
This section evaluates the SPS-lifted CT verification methodologies
from \cref{sec:application}.  We use existing CT verification tools:
the \easycrypt{} proof assistant, the \binsecrel{} symbolic execution CT
checker, and the \ctgrind{} taint analyzer.  We refer to these
combinations as SPS-\easycrypt{}, SPS-\binsec{}, and SPS-\ctgrind{}.

We evaluate the tools using an established benchmark, initially
developed by Paul Kocher for the sake of testing the Spectre mitigations in
Microsoft's C/C++ Compiler~\cite{benchmark}. This benchmark was also
used to assess the performance of SCT analysis tools in previous work,
e.g., in \citet{DBLP:conf/ndss/DanielBR21,pitchfork}.  The benchmark
contains 16 test cases that are vulnerable to Spectre-v1, of which
only four contain loops---loops are relevant for our second
methodology, which is bounded.  We consider three variants of each
case.  The first variant is the original vulnerable program (from
Kocher's benchmarks, with modifications from
Pitchfork~\cite{pitchfork}).  The second variant is patched with index
masking~\cite{indexmasking} (taken
from~\cite{DBLP:conf/ndss/DanielBR21}), where the index of each memory
access is masked with the size of the array.  The third variant is
patched with SelSLH (implemented in this work), where each value
loaded from memory is protected with an MSF\@.  Thus, we have 48
evaluation cases in total.  Of these, nineteen contain SCT violations:
the vulnerable variant of every case (sixteen total) and three
index-masked cases (these are patched with index-masking but still
insecure).

Besides this benchmark, we evaluate SPS-\easycrypt{} on four relevant
examples from the literature: the motivating example, a program that
minimizes MSF updates and cannot be verified with the type system
of~\cite{ShivakumarTyping2023}, the MAC rotate function from
MEE-CBC~\cite{FinegrainedCT}, and the Spectre-v4 example from
\cref{sec:spectrev4}.

\begin{figure}
  \centering
  \caseone{}\hspace{2em}\casefive{}
  \caption{Evaluation cases one and five.}%
  \label{fig:case-1-and-5}
\end{figure}

As an illustration of the evaluation cases, we present the first and
fifth cases from the benchmarks in \cref{fig:case-1-and-5}.
All cases share two public arrays, \jcode{pub}~and~\jcode{pub2},
a secret array,~\jcode{sec} (which is never accessed directly),
a temporary variable,~\jcode{temp} (to avoid compiler optimizations),
and an MSF,~\jcode{msf}.
The first evaluation case has no loops, and its vulnerability stems
from the bounds check in \cref{line:case1:bounds-check} that can be
speculatively bypassed.
Its index-masked and SelSLH variants protect the index of access
with a mask and an MSF, respectively.
The fifth evaluation case contains a loop in \cref{line:case5:loop}
that, when speculatively overflowed, can access the secret array.
As in the previous case, its index-masked and SelSLH variants protect
the array index.

\subsection{Noninterference Verification in SPS-\easycrypt{}}%
\label{sec:evaluation:easycrypt}

\easycrypt{}~\cite{EasyCrypt} is a proof assistant that previous work has
used as a CT verification tool~\cite{FinegrainedCT}.
As it implements Relational Hoare logic---among other features---we use
it to realize our first verification approach, presented in
\cref{sec:application:noninterference}.

\paragraph{Setup}
We manually transform each variant of each evaluation case with the
three transformations discussed in
\cref{sec:application:noninterference}, i.e., SPS, leakage
instrumentation, and assertion elimination.  We also transform two
examples of interest from the literature: a SelSLH-protected program
that cannot be verified with an SCT type system
(from~\cite{ShivakumarTyping2023}), and the MAC rotate function from
MEE-CBC~\cite{FinegrainedCT}.  \easycrypt{} readily supports our target
language---i.e., a while language with arrays and lists.  For each
case (we take case one as an example), we prove a lemma of the form
\begin{center}
  \begin{jasmincode}[
      fontsize=\footnotesize,
      lineno=false,
      frame=false,
      outerwidth=80ex
    ]
    \jasminkw{equiv} \jasmindname{case1\_sct} \jasminkw{:}
    Case1.slh\_trans \jasmintype{\(\sim\)} Case1.slh\_trans \jasmintype{:}
    phi $\land$ \jasmintype{=\{}dir\jasmintype{\} ==>}
    \jasmintype{=\{}obs\jasmintype{\}}.
  \end{jasmincode}
\end{center}
as discussed in \cref{sec:application:noninterference}.

\begin{wrapfigure}[11]{r}{0.28\linewidth}
  \centering
  \vspace{-2.5\intextsep}
  \casevoneslh{}%
  \caption{%
    The SelSLH example from~\cite{ShivakumarTyping2023} that does not
    type-check but SPS-\easycrypt{} can verify.%
  }%
  \label{fig:case-v1-slh}
\end{wrapfigure}

\paragraph{Results}
We are able to verify all secure cases from the
Kocher benchmark are SCT with average 10.28 lines of code.
(See \path{kocher.ec} in the artifact.)
In addition, we verified that three examples with SelSLH
protection are SCT, which cannot be shown to be safe with
type system.
More specifically, we considered the initialization example
from \cref{fig:init}, the example (\cref{fig:case-v1-slh})
from~\cite{ShivakumarTyping2023}, and the Spectre v4 example from \cref{fig:v4} (See \path{initialization.ec}, \path{sel_slh_typing_v1.ec}, and \path{v4.ec} in the artifact).
Verifying all three examples requires little effort:
the proof script is only 12.6 lines in average.
The reason for low proof effort is that, when the program is indeed SCT, two executions driven by the same directives are lockstep, which is naturally supported by \easycrypt{}.

\subsubsection*{MAC Rotation in MEE-CBC}
To illustrate the use of deductive verification for (S)CT verification
under fine-grained leakage models, we prove that a MAC rotation function
is CT under the cache-line model in SPS-\easycrypt{}.
We consider the MAC rotation function from the
MAC-then-Encode-then-CBC-Encrypt (MEE-CBC) authenticated encryption
scheme in TLS 1.2, which is a natural case study for CT verification
because early implementations were vulnerable to the Lucky Thirteen
timing attack~\cite{DBLP:conf/sp/AlFardanP13}.

\begin{figure}
  \begin{subfigure}[t]{0.4\linewidth}
    \centering
    \macrotinsecure{}
    \subcaption{Original Program.}%
    \label{fig:mac-rotation:origin}
  \end{subfigure}%
  \hfill%
  \begin{subfigure}[t]{0.53\linewidth}
    \centering
    \macrottranslated{}
    \subcaption{SPS and leakage Transformation.}%
    \label{fig:mac-rotation:sps}
  \end{subfigure}
  \caption{MAC rotation implementation optimized for cache-line
    leakage model.  \mdsize{}, \outvar{}, and \rotatedmac{} are public
    inputs, \secvar{} is a secret output. We assume that the following
    precondition is initially valid: $ 0 \le \secvar \le \mdsize < 64
    ~~\land~~ 64 \mid \rotatedmac$.}%
  \label{fig:mac-rotation}
\end{figure}

\Cref{fig:mac-rotation:origin} shows the code of the MAC rotation
function, which rotates a segment of memory, in which a MAC is stored.
The function takes public inputs \mdsize{}, \outvar{}, and
\rotatedmac{}, and a secret input \secvar{}, assuming the precondition
\begin{equation*}
  0 \le \secvar \le \mdsize < 64 ~~\land~~ 64 \mid \rotatedmac\text,
\end{equation*}
that is, that the secret is at most \mdsize{}, that \mdsize{} is smaller
than 64, and that \rotatedmac{} is a multiple of 64.
The function copies the contents of~\rotatedmac{} to~\outvar{}, rotated
by \secvar{}.

The MAC rotation function is CT under the cache-line leakage model with
respect to the relation
\emath{\phi \eqdef \eqset{\mdsize, \rotatedmac, \outvar}}.
The reason is that \cref{line:mac-rotation:mod} ensures that
\emath{0 \le \rovar < 64} at every loop iteration; therefore, the load
at \cref{line:mac-rotation:load} always accesses the same cache line.
Specifically, let us take \(k\) such that \emath{\rotatedmac = 64k} and
observe that the load at \cref{line:mac-rotation:load} leaks
\begin{equation*}
  \floor{\frac{\rotatedmac + \rovar}{64}}
  = \floor{\frac{64k + \rovar}{64}}
  = k + \floor{\frac{\rovar}{64}}
  = k\text,
\end{equation*}
which is a constant.
The other leaking instructions only depend on public data.

To lift our constant-time result to speculative constant-time, we
present the SPS transformation of the MAC rotation function in
\cref{fig:mac-rotation:sps}.
Since \cref{line:mac-rotation-sps:mod} is unmodified by SPS, we can use
the same argument as for CT and conclude that the original program is
\phiSCT{} under the cache-line leakage model.
We use SPS-\easycrypt{} to verify this result in
\path{mee_cbc_cache_line.ec} in the artifact.

\subsection{Assertion Safety Verification with SPS-\binsec{}}
\label{sec:evaluation:binsec}

\binsecrel{}~\cite{DBLP:conf/sp/DanielBR20} is a binary symbolic
execution engine that can efficiently verify CT\@.
As a symbolic execution engine, \binsecrel{} is a sound and
bounded-complete methodology for finding CT violations and proving their
absence (up to a given exploration depth).
It builds on the \binsec{}~\cite{DBLP:conf/wcre/DavidBTMFPM16} toolkit,
and extends \emph{Relation Symbolic Execution}
(RelSE)~\cite{DBLP:conf/ppdp/FarinaCG19} to achieve efficient CT
verification.
Intuitively, RelSE performs symbolic execution on a product program;
thus, we use \binsecrel{} to realize our second verification approach,
as outlined in \cref{sec:application:assertion-safety}.

To ensure that SPS-\binsec{} effectively lifts a CT verification tool to
SCT, we compare SPS-\binsec{} with
\binsechaunted{}~\cite{DBLP:conf/ndss/DanielBR21}. \binsechaunted{} extends
\binsecrel{} by lifting RelSE to the speculative domain (Haunted RelSE).
In a nutshell, Haunted RelSE performs aggressive pruning and sharing of
symbolic execution paths to overcome the path explosion introduced by
speculation.
We compare SPS-\binsec{} with \binsechaunted{} to evaluate the
effectiveness of our approach against a methodology tailored to SCT
verification.

\paragraph{Setup}
We manually transform each variant of each evaluation case with the two
transformations discussed in \cref{sec:application:assertion-safety},
i.e., SPS and assertion elimination.
The transformed programs are in C (see \path{kocher.c} in the artifact)
and compiled to x86 32-bit executable with the same set of compilation flags as in the Haunted RelSE
evaluation, and we encode lists as arrays together with an index to the
first element, and perform bounds checks on every list access.
We use the official Docker distributions of \binsecrel{} and
\binsechaunted{}, at versions \texttt{c417273} and
\texttt{Running RelSE}, respectively, on an Intel Core i7-10710U with
Ubuntu 20.04.
We set the exploration depth and to be 200 for both SPS-\binsec{}
and \binsechaunted{}, and the speculation window as 200 for the latter.

\begin{table}
  \centering
  \caption{%
    SPS-\binsec{}, \binsechaunted{} and SPS-\ctgrind{} evaluation results on the
    Spectre-v1 benchmarks.%
  }%
  \label{tab:eval}
  \tabspsbinsec{}
\end{table}

\paragraph{Results}
\Cref{tab:eval} presents our evaluation results for SPS-\binsec{} and
\binsechaunted{}, for the three variants of each evaluation case.
We write \evalbug{} when the tool finds an SCT violation; this is the
case for all vulnerable programs.
We write \evalok{} when the tool verifies the program as SCT\@; this is
the case for all loop-free patched programs.
Finally, we write \evalabort{} when the tool proves the absence of
violations up to the maximal exploration depth; this is the case for all
patched programs with loops.
Three of these programs, as mentioned above, are index-masked yet
contain SCT violations.
Triggering these violations, however, requires an exploration depth much
larger than the one in this evaluation.
Given these results, we can say that SPS-\binsec{} and \binsechaunted{}
agree on these evaluation cases, and therefore SPS-\binsec{} effectively
lifts \binsecrel{} to SCT on these evaluation cases.
\looseness=-1

In addition to achieving the same analysis results as \binsechaunted{},
SPS-\binsec{} reduces the verification time cost by approximately 2.6x
(43s for SPS-\binsec{} vs 111s for \binsechaunted{} int total).
This shows that the increase in code size due to SPS transformation does not handicap the verification effort.
The reason for the time difference is that SPS-\binsec{} enjoys new features of
\binsecrel{}, added since the release of \binsechaunted{}, e.g., support for
more assembly languages or improved efficiency (as shown in~\cite{Geimer2023}).

This is one of the benefits of our approach: it seamlessly inherits the
most up-to-date version of the underlying CT analysis tool with no
maintenance burden in this regard.

\subsection{Dynamic Taint Analysis with SPS-\ctgrind{}}%
\label{sec:evaluation:ctgrind}

CT-grind~\cite{ctgrind} is one of the most widely used CT analysis tools~\cite{DBLP:conf/sp/JancarFBSSBFA22}.
It performs dynamic taint analysis to detect CT violations by leveraging Valgrind, 
a popular framework for memory safety analysis.
Specifically, CT-grind marks a memory region as secret by setting the corresponding memory region as 
uninitialized and then tracks the usage of secret data during program execution.
It reports conditional branches and memory accesses that depend on this uninitialized memory region 
as potential CT violations. 
Unlike noninterference verification and assertion safety verification, 
CT-grind does not guarantee the absence of CT violations if no violations are reported.
Therefore, it can detect CT violations but not prove their absence.

\paragraph{Setup}
We reuse the SPS-transformed programs from the previous evaluation of
SPS-\binsec{}.
We mark the secret data as uninitialized memory region and run the programs with CT-grind. 
Since CT-grind is a dynamic analysis tool, we need to provide actual inputs to the programs.
Here, we randomly generate directives with a size of 231072 and input data within the range of 
zero to 200,000 for each program. We create 256 threads to run the programs in parallel 
and set a timeout of 10 minutes for each program.

\paragraph{Results}
The results are presented in the rightmost column of \cref{tab:eval}.
The legend is in the same principle for SPS-\binsec{} and \binsechaunted{}:
we write \evalbug{} when the tool finds an SCT violation, and
\evalabort{} when it finds no violations within the time limit.
SPS-\ctgrind{} reports CT violations for all vulnerable programs 
and reports no violations for patched programs.
Similar to SPS-\binsec{} and \binsechaunted{}, 
SPS-\ctgrind{} also reports no violations for the three index-masked programs 
that still contain CT violations. 
This is due to the fact that triggering these violations requires specific directives 
which takes significantly longer time to be generated randomly.
Nevertheless, we confirm that when providing the specific directives,
e.g., a list full of $\top$,
SPS-\ctgrind{} can successfully detect the CT violations in these three vulnerable 
index-masked programs.


\section{Related Work}%
\label{sec:related}
Recent surveys~\cite{DBLP:conf/sp/BarbosaBBBCLP21,CT_Survey,%
  DBLP:conf/sp/JancarFBSSBFA22,Geimer2023,SCT_Survey}, report over 50
tools for CT analysis and over 25 tools for SCT analysis.
We structure this section around three axes: microarchitectural
semantics, security notions, and verification techniques.


\subsection{Microarchitectural Semantics}
Traditionally, program semantics describe program execution at an
architectural level; concretely, these semantics define the semantics
of programs as a transition relation between architectural states, and
do not reason about microarchitectural effects of program executions.
Moreover, they typically describe a classical model of execution that
does not account for speculative or out-of-order (OOO) execution.
However, the situation has changed dramatically over the last
decade. Researchers have developed semantics that incorporate key
components of microarchitectural states, such as the cache, in the
classical model of execution~\cite{DBLP:conf/uss/DoychevFKMR13}.

\paragraph*{Speculation Window}
Several existing models of speculative execution consider a
\emph{speculation window}, during which the processor can continue
computing after making a prediction and before verifying the
correction of the prediction~\cite{spectector,pitchfork}. The notion
of speculation window reflects what happens in practice, since the
size of the reorder buffer of the CPU limits the number of
instructions that can be executed speculatively. However, it also
dictates fixing how execution backtracks once misspeculation is
detected. A common approach is to consider an \emph{always-mispredict
semantics}~\cite{spectector}. This semantics assumes that every branch
is mispredicted, and then rolled back. This semantics is intended to
maximize artificially the leakage of an execution, rather than to
provide a faithful model of execution. An advantage of this semantics
is that it captures all potential leakage in a single execution, and
dispenses from using directives. Our notion of security is stronger
than those with a speculation window: ours implies security for all
possible values of the speculation window, or even an adversarially
controlled speculation window~\cite{highassuranceSpectre}. Moreover,
allowing the security of an implementation to depend on the size of
the speculation window seems hazardous. Hence, we believe that our
notion of security is a better target for verification.

\paragraph*{Out-of-Order Execution}
Modern processors execute programs out-of-order. Although there is a
large body of work on out-of-order execution in weak memory models,
only a handful of works consider out-of-order execution in a security
context. One example is~\cite{DBLP:conf/uss/DanielBNBRP23}.  It would
be interesting to consider how to extend SPS to accommodate OOO
execution.

\paragraph*{Other Forms of Speculation}
Most models of speculative execution aim at capturing Spectre-v1 and
Spectre-v4 attacks. Other variants include Spectre-v2~\cite{spectre},
which exploit misprediction of indirect branches, and
Spectre-RSB~\cite{DBLP:conf/ccs/MaisuradzeR18,DBLP:conf/woot/KoruyehKSA18},
which exploits mispredictions of return addresses.
There are only few attempts to model and reason about these attacks
formally. Indeed, protecting against these attacks programmatically is
hard; in fact, current approaches against Spectre-v2 rely on
mitigation insertion rather than verification. These mitigations for
Spectre-v2 are based on restructuring the code to ensure that
predictions fall inside a known set of safe
targets~\cite{DBLP:conf/uss/NarayanDMCJGVSS21,DBLP:conf/ccs/ShenZOC18}.
Similarly, mitigations for Spectre-RSB are also based on
compiler
transformations~\cite{DBLP:conf/sp/MosierNMT24,DBLP:conf/asiaccs/Hetterich00R24,DBLP:conf/asplos/OlmosBCGL0SY025}.

\subsection{Security Notions and Leakage Models}
There are several alternatives to the definition of SCT\@. One notion is relative
speculative constant-time (RSCT), which requires that speculative execution of a
program does not leak more than sequential execution~\cite{Cheang2019}. Yet
another notion is leakage simulation, which involves two programs, and requires
that the execution of one (source) program does not leak more than the execution of
another (shadow) program~\cite{SCTPreservation}.
Given a source program executing speculatively in some leakage model,
there are two specific instances of leakage simulation of interest: when
the shadow program executes sequentially or speculatively---in the same
leakage model.
The first instance subsumes RCST, whereas the second instance is useful
to reason about SCT preservation.

Relative speculative constant-time can be stated as ``if the program $c$ is CT, then $c$ is SCT,'' which is a
4-property relating the sequential and speculative executions of program $c$ on
two different inputs \(i_1\)~and~\(i_2\). Existing
works~\cite{Cheang2019,spectector,Dong2023,Dongol2024,Guanciale2020} analyze
RSCT with symbolic execution, model checking, and theorem proving, upon
speculative semantics.
With the help of SPS, one can reduce RSCT to a 4-property that only considers
sequential execution. Furthermore, \citet{DBLP:conf/sp/ShivakumarBBCCGOSSY23}
shows that RSCT can be seen as a 2-property, since every speculative execution
starts with a sequential execution, i.e., the sequential executions can be
considered as prefixes of speculative executions. Therefore, RSCT can also be
defined as a 2-property relating only two sequential executions via SPS\@. It
would be interesting to investigate the equivalence between these RSCT forms.
Similar to RSCT, leakage simulation can be expressed as a 4-property, and SPS
can be used to reduce leakage simulation to a 4-property about sequential
executions. In both cases, the resulting 4-properties on sequential execution
could then be verified using a multi-program variant of Cartesian Hoare
Logic~\cite{DBLP:conf/pldi/SousaD16} or of Hyper Hoare Logic
~\cite{DBLP:journals/pacmpl/DardinierM24}, or a general form of product program.


\subsection{Security by Transformation}
As already mentioned in the introduction, the idea of using program
transformation to reduce SCT to CT has been used implicitly in the
literature to inform the design of many SCT tools.
\citet{DBLP:journals/pacmpl/BrotzmanZKT21} uses it explicitly, showing
that a program is speculative constant-time (for a flavor of SCT)
whenever its transformation is CT---however, it does not show the
converse. The transformation is more complex than ours, as the
speculative semantics is based on a speculation window.

\subsection{Verification Techniques}
In this section, we consider whether SPS can combine with different
classes of CT verification techniques and how it compares with other
approaches for ensuring SCT\@.
We remark that the answer to the first question does not follow
immediately from \cref{thm:ct-to-sct}, since this theorem proves a
reduction between semantic notions rather than applicability of
verification techniques.

\paragraph*{Deductive Verification, Symbolic Execution, and Model Checking}
  \Cref{sec:evaluation} builds on~\cite{FinegrainedCT}, which uses
  \easycrypt{} as a deductive verification method for constant-time.
  Similarly, \citet{DBLP:conf/ccs/ZinzindohoueBPB17} use F\(^\star\)
  to verify constant-time (and functional correctness) of
  cryptographic implementations. Another approach based on deductive
  verification is \citet{Coughlin2024}, which proposes a weakest
  precondition calculus for OOO speculative security and formally
  verifies it in Isabelle/HOL\@.  Their strategy is to track a pair of
  postconditions (a sequential one and a speculative one) to handle
  speculation and to use rely-guarantee reasoning to handle OOO
  execution~\cite{DBLP:conf/csfw/MantelSS11}.  Intuitively, this
  approach can be seen as applying SPS at the logical level---instead
  of transforming the program---and enjoys the versatility of Hoare
  logic.  A key difference with SPS is that it handles leakage
  explicitly, by tracking labels (public or secret) for each value,
  and aims to show an information flow property, i.e., that no secrets
  flow into leaking instructions.  This disallows semantic reasoning
  about leakage, permitted by SPS\@; in particular, it does not
  allows us to reason about mitigations as SelSLH\@.

Lightweight formal techniques like symbolic
execution~\cite{pitchfork,guanhua2020kleespectre,DBLP:conf/ndss/DanielBR21,spectector,SpecuSym}
and bounded model checking~\cite{kaibyo} have also been used for the
verification of CT and SCT\@. These tools target bounded executions,
in which every function call is inlined and every loop unrolled up to
a certain exploration depth.  On the other hand, their distinctive
advantage is that they are fully automatic and can readily find
counterexamples.  As discussed in \cref{sec:evaluation}, SPS combines
successfully with these techniques and lifts them to SCT\@. Moreover,
it solves the limitation of the speculation window.

\paragraph*{Type Systems}
Type systems are commonly used to verify constant-time and speculative
constant-time policies~\cite{blade,ShivakumarTyping2023,%
  DBLP:journals/jar/BartheBC0P20}.  Both families of type systems
consider typing judgments of the form \emath{\vdash c: \Gamma \to
  \Gamma'}, where \(\Gamma\)~and~\(\Gamma'\) are security environments
that map variables to security types.  The key difference between type
systems for constant-time and type systems for speculative
constant-time is that, in the former, the environment~\(\Gamma\) maps
every variable to a security level, while in the latter,
\(\Gamma\)~maps every variable to a pair of security levels: one for
sequential execution, and one for speculative execution.  Type systems
for speculative constant-time that support mitigations such as SLH can
be understood as a value-dependent type system, where the dependence
is relative to the misspeculation flag. Unfortunately, there is not
much work on value-dependent type systems for constant-time, hence it
is not possible to use SPS in combination for type systems for
constant-time. It would be of theoretical interest to study if one
could use SPS in combination with the instrumentation of leakage to
reduce speculative constant-time to noninterference, and use a
value-dependent information flow type system in the style
of~\cite{DBLP:conf/popl/LourencoC15} on the resulting program.

\paragraph*{Taint Analysis and Fuzzing}
Speculative constant-time verification can also be tackled with
traditional program analysis approaches like taint
analysis~\cite{Qi2021SpecTaintST,guanhua2020oo7} and
fuzzing~\cite{specfuzz,Oleksenko2022,Oleksenko2023,Jana2023}.
\cref{sec:evaluation} shows that one can effectively combine SPS with
such tools. 

\subsection{Mitigations}

Our approach allows us to verify that a program is not vulnerable to
Spectre attack after being protected by specific mitigations like
SLH\@.  In this section, we will discuss alternative approaches based
on mitigations and hardware solutions.

\paragraph*{Software-Based Mitigations}
A very common approach to protect against Spectre attacks is to
introduce mitigations in vulnerable programs.  For example, LLVM
implements a pass that introduces SLH
protections~\cite{carruth2018slh}, and
\citet{DBLP:conf/uss/ZhangBCSY23} improves on it, proposing Ultimate
SLH, to guarantee full Spectre protections. Flexible
SLH~\cite{DBLP:conf/csfw/BaumannBDHH25} further refines the approach
to reduce the number of inserted protections and formally verifies its
soundness in the Rocq prover.  These compiler transformations aim to
strengthen constant-time programs to achieve speculative
constant-time.  As such, they constitute an alternative approach to
ours, convenient when developers can afford such transformations.
Other approaches, see e.g.,~\cite{blade,DBLP:conf/sp/BulckM0LMGYSGP20,lightslh},
automatically insert fences.

\paragraph*{Hardware and Operating System Solutions}
Some proposals advocate for relying on hardware and operating system
extensions to mitigate Spectre---or lighten the software burden of
mitigating it.  For example,
ProSpeCT~\cite{DBLP:conf/uss/DanielBNBRP23} proposes a processor
design and RISC-V prototype that guarantee that CT programs are
speculatively secure against all known Spectre attacks, even under OOO
execution.  Serberus~\cite{DBLP:conf/sp/MosierNMT24} proposes a
different approach: it relies on Intel's Control-Flow Enforcement
Technology and speculation control mechanisms, on the OS to fill the
return stack buffer on context switches, and on four compiler
transformations.  These components allow Serberus to ensure that
\emph{static constant-time} programs (a strengthening of CT) are
protected against all known Spectre attacks.  These two approaches
attain much stronger guarantees than ours, and make much stronger
assumptions.  In this work, we focus on verification of software-only
mitigations.


\section{Conclusion and Future Work}%
\label{sec:conclusion}
Speculation-Passing Style is a program transformation that can be used
for speculative constant-time analysis to constant-time analysis. It
applies to Spectre-v1 and -v4 vulnerabilities, both in baseline and
fine-grained leakage models. SPS allows, without loss of precision,
to verify SCT using preexisting approaches based on deductive
verification, symbolic execution and tainting. We demonstrate the
viability of this approach using the \easycrypt{} proof assistant, the
\binsecrel{} relational symbolic execution engine, the \ctgrind{}
dynamic taint analyzer.

An important benefit of SPS is that it empowers users to verify their code in
situations for which no verification tool is readily available for a concrete
task at hand.  Potential scenarios of interest include hardware-software
contracts~\cite{DBLP:conf/sp/GuarnieriKRV21}, and in particular contracts for
future
microarchitectures~\cite{DBLP:conf/isca/VicarteSNT0KF21,DBLP:conf/ccs/BartheBCCGGRSWY24}.
We could also implement SPS at assembly level and combine it with \ctgrind{},
which may allow analyzing real-world libraries protected with LLVM SLH\@. An
alternative approach is to implement SPS at the LLVM IR and combine it with
CT-LLVM~\cite{DBLP:journals/iacr/ZhangB25}. However, as the SLH pass is in the
backend of the LLVM compiler, this approach requires moving the SLH protection pass
from the backend to the IR level, and taking care of the unintended leakages introduced
by the backend. We leave these directions for future work.

\begin{acks}
We thank the reviewers for their time and insightful feedback.
We are grateful to Benjamin Grégoire for many useful discussions on this work.
This research was supported
by the%
\grantsponsor{
  DFG%
}{
  \textit{Deutsche Forschungsgemeinschaft} (DFG, German research Foundation)%
}{}
as part of the Excellence Strategy of the German Federal and State Governments
-- \grantnum{DFG}{EXC 2092 CASA - 390781972}.
\end{acks}

\section*{Data-Availability Statement}

We provide an artifact in
\ifanon{%
  \url{https://doi.org/10.5281/zenodo.17304903}%
}{%
  \url{https://doi.org/10.5281/zenodo.17339112}%
}.
The artifact includes
\begin{itemize}
  \item Mechanized SCT proofs in \easycrypt{} proof assistant (\cref{sec:overview} and \cref{sec:evaluation});
  \item The SPS-transformed initialization example in C and the guide of how to detetect leakages by \binsec{} and \ctgrind{} (\cref{sec:overview}); and
  \item Kocher's Spectre-v1 benchmark programs in C and their SPS transformations, and the scripts for reproducing the experiment in \cref{tab:eval} (\cref{sec:evaluation}).
\end{itemize}


\bibliographystyle{ACM-Reference-Format}
\bibliography{bib}

  \ifAppendix%
    \appendix
    \newpage
    \crefalias{section}{appendix}
    \section{Sequential Semantics of the Target Language}%
\label{app:target-semantics}

\begin{figure}
  \figtgtsem{}
  \caption{Sequential semantics of the target language.}%
  \label{\labelfigtgtsem}
\end{figure}

\Cref{\labelfigtgtsem} presents the semantics of the target language.
It is the standard while-language semantics with assertions and observations.
The \reftgtsemassertT{} rule ensures that the expression evaluates to true,
while the \reftgtsemassertF{} rule ensures that if the expression evaluates to
false the program terminates with an error immediately.

\section{Soundness and Completeness of the SPS Transformation}%
\label{app:proof-trans}

To prove the correctness of the main translation, we first prepare a measure function and two lemmas.

We define $m(c)$ as how many steps the translated target program takes for the first speculative source step of a source code $c$, aligning with \cref{fig:trans}, as follows.
\begin{equation*}
  m(c) \eqdef \begin{cases}
    0 & \text{if } c \text{ is } \cnil \text,\\
    1 & \text{if } c \text{ is } \iassign{x}{e} \text,\\
    2 & \text{if } c \text{ is } \iload{x}{e} \text{ or } \istore{x}{e} \text,\\
    4 & \text{if } c \text{ is } \iif{e}{c_\top}{c_\bot} \text{ or } \iwhile{e}{c_w} \text,\\
    2 & \text{if } c \text{ is } \iinitmsf, \\
    1 & \text{if } c \text{ is } \iupdatemsf{e} \text{ or } \iprotect{x}{e} \text,\\
    m(c_1) & \text{if } c \text{ is } c_1 \capp c_2 \text.
  \end{cases}
\end{equation*}

The first lemma is a foundational building block, called step-wise simulation, saying that there is a one-to-one correspondence between a speculative source execution step and an $m(c)$-step target execution trace.

\begin{lemma}[Step-wise simulation]
\label{lemma:step-sim}
Considering any non-final source state $\st{c}{\vm}{\mem}{\ms}$ and any directive list $\seq{d}$,
the following two semantic relations are equivalent:
\begin{align*}
  \sem{\st{c}{\vm}{\mem}{\ms}}{\seq{d}}{\seq{o}}{\st{c'}{\vm'}{\mem'}{\ms'}}
\end{align*}
iff
\begin{align*}
\nsem%
  {\st{\Tprogstar{c}}{
    \msetaux{\vm}{
      \dirvar \mapsto \encode{\seq{d} \lapp \seq{d_0}},
      \msvar \mapsto \ms
    }}{\mem}{}}
  {}{T(\seq{o}, \seq{d})}
  {\st{\Tprogstar{c'}}{
    \msetaux{\vm'}{
      \dirvar \mapsto \encode{\seq{d_0}},
      \msvar \mapsto \ms'
    }}{\mem'}{}}
  {m(c)}.
\end{align*}
\end{lemma}

\begin{proof}
We will prove two directions separately, both by case analysis.

\noindent\subsection*{Source to Target}
We make an induction on the syntactic structure of $c$.

\begin{proofinductivecase}{$c$ is $c_1 \capp c_2$}
The induction hypothesis is that, if
\begin{align*}
  \sem{\st{c_1}{\vm}{\mem}{\ms}}{\seq{d}}{\seq{o}}{\st{c_1'}{\vm'}{\mem'}{\ms'}},
\end{align*}
then
\begin{align*}
  \nsem%
    {\st{\Tprogstar{c_1}}{
      \msetaux{\vm}{
        \dirvar \mapsto \encode{\seq{d} \lapp \seq{d_0}},
        \msvar \mapsto \ms
      }}{\mem}{}}
    {}{T(\seq{o}, \seq{d})}
    {\st{\Tprogstar{c_1'}}{
      \msetaux{\vm'}{
        \dirvar \mapsto \encode{\seq{d_0}},
        \msvar \mapsto \ms'
      }}{\mem'}{}}
    {m(c_1)}.
\end{align*}

In order to apply the induction hypothesis, we first analyze the source semantics.
The source execution must be justified by the source rule \refrule{fig:source-semantics}{seq},
which means that there exists $c_1'$ so that $c_1' ; c_2 = c'$ and
\begin{align*}
  \sem{\st{c_1}{\vm}{\mem}{\ms}}{\seq{d}}{\seq{o}}{\st{c_1'}{\vm'}{\mem'}{\ms'}}.
\end{align*}

Now, applying the induction hypothesis, we can obtain that 
\begin{align*}
  \nsem%
    {\st{\Tprogstar{c_1}}{
      \msetaux{\vm}{
        \dirvar \mapsto \encode{\seq{d} \lapp \seq{d_0}},
        \msvar \mapsto \ms
      }}{\mem}{}}
    {}{T(\seq{o}, \seq{d})}
    {\st{\Tprogstar{c_1'}}{
      \msetaux{\vm'}{
        \dirvar \mapsto \encode{\seq{d_0}},
        \msvar \mapsto \ms'
      }}{\mem'}{}}
    {m(c_1)},
\end{align*}
where $m(c) = m(c_1)$ by definition.

Then, with the help of the target rule \refrule{fig:target-semantics}{seq}, we can conclude that
\begin{align*}
  \nsem%
    {\st{\Tprogstar{c_1} \capp \Tprogstar{c_2}}{
      \msetaux{\vm}{
        \dirvar \mapsto \encode{\seq{d} \lapp \seq{d_0}},
        \msvar \mapsto \ms
      }}{\mem}{}}
    {}{T(\seq{o}, \seq{d})}
    {\st{\Tprogstar{c_1'} \capp \Tprogstar{c_2}}{
      \msetaux{\vm'}{
        \dirvar \mapsto \encode{\seq{d_0}},
        \msvar \mapsto \ms'
      }}{\mem'}{}}
    {m(c)},
\end{align*}
where $\Tprogstar{c_1'} \capp \Tprogstar{c_2} = \Tprogstar{c_1' \capp c_2} = \Tprogstar{c'}$.
\end{proofinductivecase}

\begin{proofbasecase}{$c$ is a single instruction or $\cskip$.}

There are nine possibilities for $c$ in total.
We will discuss four representative cases; others are similar.
For each case, we are going to show that, given the source execution reaching state $\st{\cskip}{\vm'}{\mem'}{\ms'}$ and leaking $\seq{o}$,
the target program takes exactly $m(c)$ steps to reach state $\st{\cskip}{\msetaux{\vm'}{\dirvar \mapsto \seq{d_0}, \msvar \mapsto \ms'}}{\mem'}{}$, where $c'$, $\vm'$, $\mem'$, and $\ms'$ are exactly the same as in the source execution, and producing transformed observations $T(\seq{d}, \seq{o})$.

\begin{proofcase}{$c$ is $\cskip$}
In this case, for the source execution, according to \refrule{fig:source-semantics}{refl}, the resulting state is exactly the same as the starting state and nothing is leaked.
Same for the target execution (see \refrule{fig:target-semantics}{refl}).
\end{proofcase}

\begin{proofcase}{$c$ is $\iassign{x}{e}$}
  The source execution must be an exact application of the source rule \refrule{fig:source-semantics}{assign}:
  \begin{align*}
    \st{\iassign{x}{e}}{\vm}{\mem}{\ms}
    \step{\dstep}{\onone}
    \st{\cskip}{\mset{\vm}{x}{\eval{e}{\vm}}}{\mem}{\ms},
  \end{align*}
  which indicates $\seq{d}$ must be $\epsilon$, and in the resulting state, $\vm' = \msetaux{\vm}{\dirvar \mapsto \eval{e}{\vm}}$, $\mem' = \mem$, $\ms' = \ms$, and the observations produced $\seq{o} = \epsilon$.

  Let us analyze how the target program runs by $m(c) = 1$ steps.
  The target execution is the application of the target rule \refrule{fig:target-semantics}{assign}:
  \begin{gather*}
    \st
      {\iassign{x}{e}}
      {
        \msetaux
          {\vm}
          {
            \dirvar \mapsto \encode{\seq{d} \lapp \seq{d_0}},
            \msvar \mapsto \ms
          }
      }
      {\mem}
      {}
    \step{}{\onone}
    \st
      {\cskip}
      {
        \msetaux
          {\vm}
          {
            \dirvar \mapsto \encode{\seq{d_0}},
            \msvar \mapsto \ms,
            x \mapsto \eval{e}{\vm}
          }
      }
      {\mu}
      {},
  \end{gather*}
  which means that $\vm' = \msetaux{\vm}{\dirvar \mapsto \eval{e}{\vm}}$, $\mem' = \mem$, $\ms' = \ms$ -- exactly aligned with the source resulting state.
  Also, the target observations can be obtained from the source observations through the transformer: $\lnil = T(\lnil, \lnil) = T(\seq{d}, \seq{o})$.
\end{proofcase}

The cases for $\iload{x}{e}, \istore{e}{x}, \iupdatemsf{e}$, and $\iprotect{x}{e}$ are similar to this case.

\begin{proofcase}{$c$ is $\iif{e}{c_\top}{c_\bot}$}
  The source execution must be an exact application of the source rule \refrule{fig:source-semantics}{cond}:
  \begin{align*}
      \st{\iif{e}{c_\top}{c_\bot}}{\vm}{\mem}{\ms}
      \step{[\dforce{b}]}{[\obranch{b'}]}
      \st{c_b}{\vm}{\mem}{\ms \lor (b \not= b')},
  \end{align*}
  where $b' = \eval{e}{\vm}$,
  which indicates $\seq{d}$ must be $[\dforce{b}]$, and in the resulting state, $\vm' = \vm$, $\mem' = \mem$, $\ms' = \ms \lor (b \not= b')$, and the observations produced $\seq{o} = [\obranch{b'}]$.

  Now, let us analyze how the target program runs by $m(c) = 4$ steps.
  The target program $\Tprogstar{c}$ is
  \begin{equation*}
    \begin{aligned}[t]
      & \mathtt{if}~(e)~\{~\} \\
      & \mathtt{if}~(\hd{\dirvar}) \{ \\
      & \quad \iassign{\dirvar}{\tl{\dirvar}}\capp \\
      & \quad \iassign{\msvar}{\msvar\vee \neg~e}\capp \\
      & \quad \Tprogstar{c_\top} \\
      & \} ~\mathtt{else}~ \{ \\
      & \quad \iassign{\dirvar}{\tl{\dirvar}}\capp \\
      & \quad \iassign{\msvar}{\msvar\vee e}\capp \\
      & \quad \Tprogstar{c_\bot} \\
      & \}
    \end{aligned}
  \end{equation*}
  
  The target program starts at $\st{\Tprogstar{c}}{\msetaux{\vm}{\dirvar \mapsto \dforce{b} \lapp \seq{d_0}, \msvar \mapsto \ms}}{\mem}{}$.
  First, take a step by applying the rule \refrule{fig:target-semantics}{cond}, which leaks $\obranch{\eval{e}{\vm}}$, the same as the source execution.
  Second, take again a step by applying the rule \refrule{fig:target-semantics}{cond}, which leaks $\obranch{d}$.
  Third, take a step by applying the rule \refrule{fig:target-semantics}{assign}, which updates $\dirvar$ with $\tl{\dirvar}$.
  Fourth, take a step by applying the rule \refrule{fig:target-semantics}{assign}, which updates $\msvar$ with $\ms \lor (b \not= b')$.
  Finally, we've reached a state $\st{\cskip}{\msetaux{\vm}{\dirvar \mapsto \seq{d_0}, \msvar \mapsto \ms \lor (b \not= b')}}{\mem}{}$,
    which means that $\vm' = \vm$, $\mem' = \mem$, and $\ms' = \ms \lor (b \not= b')$ -- exactly aligned with the source resulting state.
  Also, the observations during these step can be obtained from the source observations through the transformer:
    $[\obranch{b'}, \obranch{b}] = T([\obranch{b'}], [\dforce{b}]) = T(\seq{o}, \seq{d})$.

  The case for $\iwhile{e}{c'}$ is similar.
\end{proofcase}

\begin{proofcase}{$c$ is $\iinitmsf$}
  The source execution must be an exact application of the source rule \refrule{fig:source-semantics}{init-msf}:
  \begin{align*}
    \st{\iinitmsf}{\vm}{\mem}{\bot}
    \step{\dstep}{\onone}
    \st{\cskip}{\msetaux{\vm}{\msf \mapsto \bot}}{\mem}{\bot},
  \end{align*}
  which indicates $\seq{d}$ must be $\epsilon$, and in the resulting state, $\vm' = \vm$, $\mem' = \mem$, $\ms' = \bot$, and the observations produced $\seq{o} = \epsilon$.

  Let us analyze how the target program runs by $m(c) = 2$ step.
  The target execution starts at
  $$\st{\iassert{\neg \msvar} \capp \iassign{\msf}{\msfnomask}}{\msetaux{\vm}{\dirvar \mapsto \seq{d_0}, \msvar \mapsto \bot}}{\mem}{}.$$
  First, take a step by the rule \reftgtsemassertT{}, which passes the assertion as $\vm(\msvar) = \bot$,
  Second, take a step by the rule \refrule{fig:target-semantics}{assign}, which updates $\msf$ by $\msfnomask$.
  Finally, we've reached a state $$\st{\cskip}{\msetaux{\vm}{\dirvar \mapsto \seq{d_0}, \msvar \mapsto \bot, \msf \mapsto \msfnomask}}{\mem}{},$$
    which means that $\vm' = \vm$, $\mem' = \mem$, and $\ms' = \bot$ -- exactly aligned with the source resulting state.
  Also, the target observations can be obtained from the source observations through the transformer: $\lnil = T(\lnil, \lnil) = T(\seq{d}, \seq{o})$.
\end{proofcase}
\end{proofbasecase}

\noindent\subsection*{Target to Source}
We make an induction on the syntactic structure of $c$, as the other direction.

\begin{proofinductivecase}{$c$ is $c_1 \capp c_2$}
The induction hypothesis is that, for any observations $\seq{o}$ and directives $\seq{d}$, if
\begin{align*}
  \nsem%
    {\st{\Tprogstar{c_1}}{
      \msetaux{\vm}{
        \dirvar \mapsto \encode{\seq{d} \lapp \seq{d_0}},
        \msvar \mapsto \ms
      }}{\mem}{}}
    {}{T(\seq{o}, \seq{d})}
    {\st{\Tprogstar{c_1'}}{
      \msetaux{\vm'}{
        \dirvar \mapsto \encode{\seq{d_0}},
        \msvar \mapsto \ms'
      }}{\mem'}{}}
    {m(c_1)},
\end{align*}
then
\begin{align*}
  \sem{\st{c_1}{\vm}{\mem}{\ms}}{\seq{d}}{\seq{o}}{\st{c_1'}{\vm'}{\mem'}{\ms'}}.
\end{align*}

In order to apply the induction hypothesis, we first analyze the target semantics.
The target execution must be derived by the source rule \refrule{fig:source-semantics}{seq},
which means that there exists $c_1'$ so that $c_1' ; c_2 = c'$ and $\Tprogstar{c_1' ; c_2} = \Tprogstar{c_1'} ; \Tprogstar{c_2}$ and
\begin{align*}
  \nsem%
    {\st{\Tprogstar{c_1}}{
      \msetaux{\vm}{
        \dirvar \mapsto \encode{\seq{d} \lapp \seq{d_0}},
        \msvar \mapsto \ms
      }}{\mem}{}}
    {}{T(\seq{o}, \seq{d})}
    {\st{\Tprogstar{c_1'}}{
      \msetaux{\vm'}{
        \dirvar \mapsto \encode{\seq{d_0}},
        \msvar \mapsto \ms'
      }}{\mem'}{}}
    {m(c_1)},
\end{align*}
where $m(c) = m(c_1)$ by definition.

Now, applying the induction hypothesis, we can obtain that 
\begin{align*}
  \sem{\st{c_1}{\vm}{\mem}{\ms}}{\seq{d}}{\seq{o}}{\st{c_1'}{\vm'}{\mem'}{\ms'}}.
\end{align*}

Then, with the help of the source rule \refrule{fig:target-semantics}{seq}, we can conclude that
\begin{align*}
  \sem{\st{c_1 \capp c_2}{\vm}{\mem}{\ms}}{\seq{d}}{\seq{o}}{\st{c_1' \capp c_2}{\vm'}{\mem'}{\ms'}}.
\end{align*}
\end{proofinductivecase}

\begin{proofbasecase}{$c$ is a single instruction or $\cskip$.}

There are nine possibilities for $c$ in total.
We will discuss four representative cases; others are similar.
For each case, we are going to show that, for any directives $\seq{d}$ and observations $\seq{o}$,
if there is a target execution
\begin{align*}
  \nsem%
    {\st{\Tprogstar{c}}{
      \msetaux{\vm}{
        \dirvar \mapsto \encode{\seq{d} \lapp \seq{d_0}},
        \msvar \mapsto \ms
      }}{\mem}{}}
    {}{T(\seq{o}, \seq{d})}
    {\st{\Tprogstar{c'}}{
      \msetaux{\vm'}{
        \dirvar \mapsto \encode{\seq{d_0}},
        \msvar \mapsto \ms'
      }}{\mem'}{}}
    {m(c)},
\end{align*}
then, there is a corresponding one-step source execution
\begin{align*}
  \sem{\st{c}{\vm}{\mem}{\ms}}{\seq{d}}{\seq{o}}{\st{c'}{\vm'}{\mem'}{\ms'}},
\end{align*}
where $c'$, $\vm'$, $\mem'$, and $\ms'$ are exactly the same in source and target.

\begin{proofcase}{$c$ is $\cskip$}
In this case, for the target execution, according to the target rule \reftgtsemrefl{},
the resulting state is exactly the same as the starting state and nothing is leaked.
Same for the source execution (see the source rule \refsemrefl{}).
\end{proofcase}

\begin{proofcase}{$c$ is $\iassign{x}{e}$}
  Let us execute the target program by $m(c) = 1$ step.
  The target execution is an application of the target rule \reftgtsemassign{}:
  \begin{gather*}
    \st
      {\iassign{x}{e}}
      {
        \msetaux
          {\vm}
          {
            \dirvar \mapsto \encode{\seq{d} \lapp \seq{d_0}},
            \msvar \mapsto \ms
          }
      }
      {\mem}
      {}
    \step{}{\onone}
    \st
      {\cskip}
      {
        \msetaux
          {\vm}
          {
            \dirvar \mapsto \encode{\seq{d_0}},
            \msvar \mapsto \ms,
            x \mapsto \eval{e}{\vm}
          }
      }
      {\mu}
      {},
  \end{gather*}
  which means that $\vm' = \msetaux{\vm}{\dirvar \mapsto \eval{e}{\vm}}$, $\mem' = \mem$, $\ms' = \ms$.
  
  Now, turn to the source execution.
  The source execution is an application of the source rule \refsemassign{}:
  \begin{align*}
    \st{\iassign{x}{e}}{\vm}{\mem}{\ms}
    \step{\dstep}{\onone}
    \st{\cskip}{\mset{\vm}{x}{\eval{e}{\vm}}}{\mem}{\ms},
  \end{align*}
  which indicates $\seq{d}$ must be $\epsilon$, and in the resulting state,
  $\vm' = \msetaux{\vm}{\dirvar \mapsto \eval{e}{\vm}}$, $\mem' = \mem$, $\ms' = \ms$ -- exactly aligned with the target resulting state.
  The observations produced $\seq{o} = \epsilon$, which means that the relation between source and target observations satisfies the transformer:
  $\lnil = T(\lnil, \lnil) = T(\seq{d}, \seq{o})$.
\end{proofcase}

The cases for $\iload{x}{e}, \istore{e}{x}, \iupdatemsf{e}$, and $\iprotect{x}{e}$ are similar to this case.

\begin{proofcase}{$c$ is $\iif{e}{c_\top}{c_\bot}$}
  The target program $\Tprogstar{c}$ is
  \begin{equation*}
    \begin{aligned}[t]
      & \mathtt{if}~(e)~\{~\} \\
      & \mathtt{if}~(\hd{\dirvar}) \{ \\
      & \quad \iassign{\dirvar}{\tl{\dirvar}}\capp \\
      & \quad \iassign{\msvar}{\msvar\vee \neg~e}\capp \\
      & \quad \Tprogstar{c_\top} \\
      & \} ~\mathtt{else}~ \{ \\
      & \quad \iassign{\dirvar}{\tl{\dirvar}}\capp \\
      & \quad \iassign{\msvar}{\msvar\vee e}\capp \\
      & \quad \Tprogstar{c_\bot} \\
      & \}
    \end{aligned}
  \end{equation*}

  Let us execute the target program by $m(c) = 4$ steps.
  The execution starts at $$\st{\Tprogstar{c}}{\msetaux{\vm}{\dirvar \mapsto \dforce{b} \lapp \seq{d_0}, \msvar \mapsto \ms}}{\mem}{}.$$
  First, take a step by applying the rule \refrule{fig:target-semantics}{cond}, which leaks $\obranch{\eval{e}{\vm}}$, the same as the source execution.
  Second, take again a step by applying the rule \refrule{fig:target-semantics}{cond}, which leaks $\obranch{d}$.
  Third, take a step by applying the rule \refrule{fig:target-semantics}{assign}, which updates $\dirvar$ with $\tl{\dirvar}$.
  Fourth, take a step by applying the rule \refrule{fig:target-semantics}{assign}, which updates $\msvar$ with $\ms \lor (b \not= b')$.
  Finally, we've reached a state $\st{\cskip}{\msetaux{\vm}{\dirvar \mapsto \seq{d_0}, \msvar \mapsto \ms \lor (b \not= b')}}{\mem}{}$,
    which means that $\vm' = \vm$, $\mem' = \mem$, and $\ms' = \ms \lor (b \not= b')$.
  The observations produced during these steps are $[\obranch{b'}, \obranch{b}]$.

  Turn to the source program.
  The source execution is an application of the source rule \refrule{fig:source-semantics}{cond}:
  \begin{align*}
      \st{\iif{e}{c_\top}{c_\bot}}{\vm}{\mem}{\ms}
      \step{[\dforce{b}]}{[\obranch{b'}]}
      \st{c_b}{\vm}{\mem}{\ms \lor (b \not= b')},
  \end{align*}
  where $b' = \eval{e}{\vm}$,
  which indicates $\seq{d}$ must be $[\dforce{b}]$, and in the resulting state, $\vm' = \vm$, $\mem' = \mem$, $\ms' = \ms \lor (b \not= b')$
  -- exactly aligned with the target resulting state.
  The observations produced are $\seq{o} = [\obranch{b'}]$, which satisfies the relation of $T$ with the target observations:
    $[\obranch{b'}, \obranch{b}] = T([\obranch{b'}], [\dforce{b}]) = T(\seq{o}, \seq{d})$.

  The case for $\iwhile{e}{c'}$ is similar.
\end{proofcase}

\begin{proofcase}{$c$ is $\iinitmsf$}
  As the target program can execute $m(c) = 2$ steps, we first show that the value of $\msvar$ in the starting state must be $\bot$.
  This is because $\Tprogstar{c}$ is $\iassert{\neg \msvar} \capp \iassign{\msf}{\msfnomask}$,
  if $\msvar$ is not $\bot$, then the execution will reach the error state and halt within one step by the rule \reftgtsemassertF{}.

  Let us run the target program by $m(c) = 2$ steps now.
  The target execution starts at
  $$\st{\iassert{\neg \msvar} \capp \iassign{\msf}{\msfnomask}}{\msetaux{\vm}{\dirvar \mapsto \seq{d_0}, \msvar \mapsto \bot}}{\mem}{}.$$
  First, take a step by the rule \reftgtsemassertT{}, which passes the assertion as $\vm(\msvar) = \bot$,
  Second, take a step by the rule \refrule{fig:target-semantics}{assign}, which updates $\msf$ by $\msfnomask$.
  Finally, we've reached a state $$\st{\cskip}{\msetaux{\vm}{\dirvar \mapsto \seq{d_0}, \msvar \mapsto \bot, \msf \mapsto \msfnomask}}{\mem}{},$$
    which means that $\vm' = \vm$, $\mem' = \mem$, and $\ms' = \bot$.
  No observations are produced during these steps.

  Turn to the source program.
  The source execution must be an exact application of the source rule \refrule{fig:source-semantics}{init-msf}:
  \begin{align*}
    \st{\iinitmsf}{\vm}{\mem}{\bot}
    \step{\dstep}{\onone}
    \st{\cskip}{\msetaux{\vm}{\msf \mapsto \bot}}{\mem}{\bot},
  \end{align*}
  which indicates $\seq{d}$ must be $\epsilon$, and in the resulting state, $\vm' = \vm$, $\mem' = \mem$, $\ms' = \bot$
    -- exactly aligned with the target resulting state.
  The observations produced are $\seq{o} = \epsilon$, which satisfies the relation of $T$ together with the target observations:
    $\lnil = T(\lnil, \lnil) = T(\seq{d}, \seq{o})$.
\end{proofcase}
\end{proofbasecase}

\end{proof}

The second lemma is not about soundness, but only about completeness of SPS, i.e., a target trace has a corresponding source trace.
This lemma is called target decomposition, saying that a target trace can be split into two traces end-to-end,
where the first trace corresponds to a speculative source step and the second trace will be split inductively.

\begin{lemma}[Target decomposition]
\label{lemma:target-decomp}
If we have such a target trace,
\begin{align*}
\sem*%
  {\st{\Tprogstar{c}}{\vm_t}{\mem_t}{}}
  {}{\seq{o}}
  {s'} \land final(s'),
\end{align*}
then one of the three possibilities holds:
(1) $c$ is $\cskip$,
or (2) the first instruction of $c$ is $\iinitmsf$ and $\rho_t(\msvar) = \top$,
or (3) there exists an intermediate code $c_{m} \not= c$ s.t. the whole trace can be split into two traces, as follows,
$$
\nsem%
  {\st{\Tprogstar{c}}{\vm_t}{\mem_t}{}}
  {}{\seq{o}_1}
  {\st{\Tprogstar{c_{m}}}{\vm_{m, t}}{\mem_{m, t}}{}}
  {m(c)}
$$
and
$$
\sem*%
  {
    \st
      {\Tprogstar{c_{m}}}
      {\vm_{m, t}}
      {\mem_{m, t}}
      {}
  }
  {}{\seq{o}_2}
  {s'},$$
where $\seq{o} = \seq{o}_1 \lapp \seq{o}_2$.
\end{lemma}
\begin{proof}
The proof is by a case analysis on the first instruction of $c$.
For each case, we will show the existence of the first trace
$\nsem%
{\st{\Tprogstar{c}}{\vm_t}{\mem_t}{}}
{}{\seq{o_1}}
{\st{\Tprogstar{c_{m}}}{\vm_{m, t}}{\mem_{m, t}}{}}
{m(c)}$.
And the existence of the second trace will be established automatically.
Because the target semantics is deterministic, if we know from state \st{\Tprogstar{c}}{\vm_t}{\mem_t}{}, the program terminates;
and the program can execute $m(c)$ steps reaching \st{\Tprogstar{c_m}}{\vm_{m, t}}{\mem_{m, t}}{}, then we can know that this execution is a prefix of the terminating trace, and the execution from \st{\Tprogstar{c_m}}{\vm_{m, t}}{\mem_{m, t}}{} also terminates.

\begin{proofcase}{$\cnil$}
This case directly holds.
\end{proofcase}

\begin{proofcase}{$\iif{e}{c_\top}{c_\bot}$}
In this case, $\Tprogstar{c}$ is
\begin{equation*}
\begin{aligned}[t]
  & \mathtt{if}~(e)~\{~\} \capp \\
  & \mathtt{if}~(\dirsem{\hd{\dirvar}}{e})~\{ \\
  & \quad \iassign{\dirvar}{\tl{\dirvar}}\capp \\
  & \quad \iassign{\msvar}{\msvar\vee \neg~e}\capp \\
  & \quad \Tprogstar{c_\top} \\
  & \} ~\mathtt{else}~ \{ \\
  & \quad \iassign{\dirvar}{\tl{\dirvar}}\capp \\
  & \quad \iassign{\msvar}{\msvar\vee e}\capp \\
  & \quad \Tprogstar{c_\bot} \\
  & \} \capp \\
  & \Tprogstar{c_0}
\end{aligned}
\end{equation*}
After $m(c) = 4$ steps -- \refrule{fig:target-semantics}{cond}, \refrule{fig:target-semantics}{cond}, \refrule{fig:target-semantics}{assign}, and \refrule{fig:target-semantics}{assign} -- we will reach $\Tprogstar{c_\top}$ or $\Tprogstar{c_\bot}$.
Thus, we know the existence of $c_m$: either $c_\top$ or $c_\bot$. $\vm_{m, t}$ and $\mem_{m, t}$ can be computed by applying the four rules.
\end{proofcase}
  
\begin{proofcase}{$\iwhile{e}{c_w}$}
In this case, $\Tprogstar{c}$ is
\begin{equation*}
  \begin{aligned}[t]
    & \mathtt{if}~(e)~\{~\} \capp \\
    & \mathtt{while}~(\hd{\dirvar})~\{\\
    & \quad \iassign{\dirvar}{\tl{\dirvar}}\capp  \\
    & \quad \iassign{\msvar}{\msvar\vee \neg e}\capp \\
    & \quad \Tprogstar{c_w}\capp \\
    & \}\capp \\
    & \iassign{\msvar}{\msvar\vee e}\capp \\
    & \iassign{\dirvar}{\tl{\dirvar}}\capp \\
    & \Tprogstar{c_0}
  \end{aligned}
\end{equation*}
After $m(c) = 4$ steps -- \refrule{fig:target-semantics}{cond}, \refrule{fig:target-semantics}{while}, \refrule{fig:target-semantics}{assign}, and \refrule{fig:target-semantics}{assign} -- we will reach $\Tprogstar{c_w}$ or $\Tprogstar{c_0}$.
Thus, we know the existence of $c_m$: either $c_w$ or $c_0$. $\vm_{m, t}$ and $\mem_{m, t}$ can be computed by applying the four rules.
\end{proofcase}
  
\begin{proofcase}{$\iinitmsf$}
In this case, $\Tprogstar{c}$ is $\iassert{\neg \msvar} \capp \iassign{\msf}{\msfnomask}$.
If $\vm_t(\msvar) = \top$, then we reach $\cnil$ in one step by \reftgtsemassertF{}, without any change in the register file or memory.
This reflects the first possible case in the lemma statement:
$c$ is of the form $\iinitmsf; c'$, and $\vm_t(\msvar) = \top$, $\vm_t = \vm_t'$ and $\mem_t = \mem_t'$.
Otherwise, when $\vm_t(\msvar) = \bot$, we will also take $m(c) = 2$ steps -- \reftgtsemassertT{} and \reftgtsemassign{} -- until reaching $\Tprogstar{c_0}$.
Thus, when $\vm_t(\msvar) \not= \top$, we can take $c_0$ for $c_m$.
\end{proofcase}

\begin{proofcases}{Others}
For other cases, after $m(c)$ steps, we will reach $\Tprogstar{c_0}$. Thus, we can take $c_0$ for $c_m$.
\end{proofcases}
\end{proof}

\begin{theorem}[Soundness and completeness of SPS]%
\label{app:thm:soundness-completeness}
  \begin{gather*}
  \exists c',
  \sem*{c(i)}{\seq{d}}{\seq{o}}{\st{c'}{\vm'}{\mem'}{\ms'}} \land final(\st{c'}{\vm'}{\mem'}{\ms'})
  \end{gather*}
  \begin{center}iff\end{center}
  \begin{gather*}
  \left(
  \sem*%
  {\Tprog{c}(i, \seq{d} \lapp \seq{d_0})}
  {}{T(\seq{o},\seq{d})}
  {\st{\cnil}{\msetaux{\vm'}{\dirvar \mapsto \seq{d_0}, \msvar \mapsto \ms}}{\mem'}{}}
  \quad\textrm{ or }\quad
  \sem*%
  {\Tprog{c}(i, \seq{d} \lapp \seq{d_0})}
  {}{T(\seq{o},\seq{d})}
  {\sterror}
  \right).
  \end{gather*}
\end{theorem}

\begin{proof}
We will inductively prove the following statement:
$$\sem*{\st{c}{\vm}{\mem}{\ms}}{\seq{d}}{\seq{o}}{\st{c'}{\vm'}{\mem'}{\ms'}} \land final(\st{c'}{\vm'}{\mem'}{\ms'})$$
\begin{center}iff\end{center}
\[
\mathmakebox[\linewidth][c]{%
\left(
\begin{gathered}
  \sem*%
 {\st{\Tprogstar{c}}%
 {\msetaux{\vm}{\dirvar \mapsto \seq{d} \lapp \seq{d_0}, \msvar \mapsto \ms}}%
 {\mem_t}{}}
 {}{T(\seq{o},\seq{d})}
 {\st{\Tprogstar{c'}}%
 {\msetaux{\vm'}{\dirvar \mapsto \seq{d_0}, \msvar \mapsto \ms'}}%
 {\mem'}{}} \\[0.5em]
  \text{or}\quad
  \sem*%
 {\st{\Tprogstar{c}}{\msetaux{\vm}{\dirvar \mapsto \seq{d} \lapp \seq{d_0}, \msvar \mapsto \ms}}{\mem}{}}
 {}{T(\seq{o},\seq{d})}
 {\sterror}
\end{gathered}
\right)
}
\]

This statement can derive the theorem statement in two steps:
(1) instantiate $\ms$ as $\top$,
(2) take care of the initialization statement of $\ms$ at the beginning of $\Tprog{c}$.



We will prove two directions of this statement separately.

\subsection*{Source to Target (Soundness)}
To prove the soundness, we make an induction on the execution steps, i.e., multistep execution rules of the speculative semantics (\refsemrefl{} and \refsemtrans{}).

\begin{proofbasecase}{the starting state is already final, i.e., the source execution can only be derived by \refsemrefl.}
In this case, $c$ is $\cnil$ or (the first instruction of $c$ is $\iinitmsf$ and $\ms = \top$). We consider these two sub-cases separately.
\begin{proofsubcase}{$c$ is $\cnil$}
$\Tprogstar{\cnil}$ is still $\cnil$, then, directly from the target rule \reftgtsemrefl{}, we have
\begin{gather*}
  \sem*
  {\st{\Tprog{\cnil}}{\msetaux{\vm}{\dirvar \mapsto \seq{d_0}, \msvar \mapsto \ms}}{\mem}{}}
  {}{\lnil}
  {\st{\cnil}{\vm}{\mem}{}}.
\end{gather*}
\end{proofsubcase}

\begin{proofsubcase}{$c$ is $\iinitmsf$ and $\ms = \top$.}
$\Tprogstar{\iinitmsf}$ is $\iassert{\neg \msvar} ; \iassign{\msf}{\bot}$.
Then, by the target rule \reftgtsemassertT,
we have target execution
\begin{gather*}
\sem*%
  {\st{\Tprogstar{\iinitmsf}}{\msetaux{\vm}{\dirvar \mapsto \seq{d_0}, \msvar \mapsto \top}}{\mem}{}}
  {}{\lnil}
  {\sterror}.
\end{gather*}
\end{proofsubcase}
\end{proofbasecase}

\begin{proofinductivecase}{the starting state is not final, i.e., the source execution can be derived by \refsemtrans.}
By the rule \refsemtrans{} of the source language,
the source execution can be split as follows, the first step and the rest:
\begin{gather*}
\st{c}{\vm}{\mem}{\ms}
\xrightarrow[\seq{d_1}]{\seq{o_1}}\mathrel{\vphantom{\to}}
\st{c''}{\vm''}{\mem''}{\ms''}
\text{ and }
\st{c''}{\vm''}{\mem''}{\ms''}
\xrightarrow[\seq{d_2}]{\seq{o_2}}\mathrel{\vphantom{\to}^\ast}
\st{c'}{\vm'}{\mem'}{\ms'},
\end{gather*}
where $\seq{o} = \seq{o_1} \lapp \seq{o_2}$ and $\seq{d} = \seq{d_1} \lapp \seq{d_2}$.

We will convert these two executions from source to target separately.
By applying \cref{lemma:step-sim} to the first step, we obtain
\begin{align*}
\sem*%
  {\st
    {\Tprogstar{c}}
    {\msetaux{\vm}{
      \dirvar \mapsto \seq{d_1} \lapp \seq{d_2} \lapp \seq{d_0},
      \msvar \mapsto \ms
      }}
    {\mem}
    {}
  }
  {}{T(\seq{o_1}, \seq{d_1})}
  {\st
    {\Tprogstar{c''}}
    {\msetaux{\vm''}{
      \dirvar \mapsto \seq{d_2} \lapp \seq{d_0},
      \msvar \mapsto \ms''
    }}
    {\mem''}
    {}
  }.
\end{align*}
By the induction hypothesis, we have that
\begin{align*}
\sem*%
  {\st
    {\Tprogstar{c''}}
    {\msetaux{\vm''}{
      \dirvar \mapsto \seq{d_2} \lapp \seq{d_0},
      \msvar \mapsto \ms''}}
    {\mem''}
    {}
  }
  {}{T(\seq{o_2}, \seq{d_2})}
  {\st
    {\cnil}
    {\msetaux{\vm'}{
      \dirvar \mapsto \seq{d_0},
      \msvar \mapsto \ms'}}
    {\mem'}
    {}
  }.
\end{align*}

Combining the two target traces above, by the target rule \reftgtsemtrans{},
we can conclude that we have a target execution almost in the form that we want,
\begin{align*}
  \sem*%
  {\st
    {\Tprogstar{c}}
    {\msetaux{\vm}{
      \dirvar \mapsto \seq{d_1} \lapp \seq{d_2} \lapp \seq{d_0},
      \msvar \mapsto \ms
      }}
    {\mem}
    {}}
  {}{T(\seq{o_1}, \seq{d_1}) \lapp T(\seq{o_2}, \seq{d_2})}
  {\st
    {\cnil}
    {\msetaux{\vm'}{
      \dirvar \mapsto \seq{d_0},
      \msvar \mapsto \ms'
    }}
    {\mem'}
    {}}.
\end{align*}
The only difference of this semantic relation and the target relation of interest is the observation list.
Luckily, this gap is not difficult to fill.
According to the definition of $T$ (intuitively, inserting directives after branch observations), we have $T(\seq{o_1}, \seq{d_1}) \lapp T(\seq{o_2}, \seq{d_2}) = T(\seq{o_1} \lapp \seq{o_2}, \seq{d_1} \lapp \seq{d_2})$.
\end{proofinductivecase}

\subsection*{Target to Source (Completeness)}

We will make an induction on the number of the execution steps of the target execution, to show that, if we have
\[
\mathmakebox[\linewidth][c]{%
\left(
\begin{gathered}
  \sem*%
  {\st{\Tprogstar{c}}%
    {\msetaux{\vm}{\dirvar \mapsto \seq{d} \lapp \seq{d_0}, \msvar \mapsto \ms}}%
    {\mem_t}{}}
  {}{\seq{o}}
  {\st{\cnil}%
    {\msetaux{\vm'}{\dirvar \mapsto \seq{d_0}, \msvar \mapsto \ms'}}%
    {\mem'}{}}
  \\[0.5em]
  \text{or}\quad
  \sem*%
  {\st{\Tprogstar{c}}{\msetaux{\vm}{\dirvar \mapsto \seq{d} \lapp \seq{d_0}, \msvar \mapsto \ms}}{\mem}{}}
  {}{T(\seq{o},\seq{d})}
  {\sterror}
\end{gathered}
\right)
}
\]
then
$$
\exists c', \sem*{\st{c}{\vm}{\mem}{\ms}}{\seq{d}}{T^{-1}(\seq{o})}{\st{c'}{\vm'}{\mem'}{\ms'}} \land final(\st{c'}{\vm'}{\mem'}{\ms'}),
$$
where $T^{-1}(\seq{o})$ is the sequence obtained from $\seq{o}$ by removing the second branch observation for each adjacent branch observation pair.

\begin{proofbasecase}{the execution step is zero, i.e., the target execution is directly derived by \reftgtsemrefl{}}
In this case, we can know that $c$ is $\cnil$ and $\seq{d} = \lnil$, $\vm = \vm'$, $\mem = \mem'$, $\seq{o} = \lnil$, $\ms = \ms'$.
Then, from the source rule \refsemrefl{}, we get
\begin{align*}
  \sem*
    {
      \st
        {\cnil}
        {\vm}
        {\mem}
        {\ms}
    }
    {\lnil}
    {\lnil}
    {\st{\cnil}{\vm'}{\mem'}{\ms}}.
\end{align*}
\end{proofbasecase}

\begin{proofinductivecase}{the target execution takes more than zero step, i.e., the target execution is derived by \reftgtsemtrans{}}

By applying the target decomposition lemma (\cref{lemma:target-decomp}) onto the target execution,
we obtain three possibilities.

\begin{proofcase}{$c$ is $\cnil$}
We can get the expected source execution by the source rule \refsemrefl{}.
\end{proofcase}

\begin{proofcase}{$c$ is $\iinitmsf$ and $\vm_t(\ms) = \top$}
The target execution only applies one rule -- \reftgtsemassertT{} -- on the transformed program: $\iassert{\neg \msvar} \capp \iassign{\msf}{\msfnomask}$.
Thus, from the target execution, we can know that $\seq{d} = \seq{o} = \lnil$ and $\ms = \ms' = \top$.
Then, we can get the source execution from the source rule \refsemrefl{},
\begin{align*}
  \sem*
    {
      \st
        {\iinitmsf}
        {\vm}
        {\mem}
        {\top}
    }
    {\lnil}
    {\lnil}
    {
      \st
        {\iinitmsf}
        {\vm'}
        {\mem'}
        {\top}
    }.
\end{align*}
\end{proofcase}

\begin{proofcase}{the target execution is composable as two traces}
We can split the target execution into two target traces, connected end-to-end, as follows.
$$
\nsem%
  {
    \st{\Tprogstar{c}}{\msetaux{\vm}{\dirvar \mapsto \seq{d} \lapp \seq{d_0}, \msvar \mapsto \ms}}{\mem}{}
  }
  {}{\seq{o_1}}
  {
    \st{\Tprogstar{c_{m}}}{\vm_{m, t}}{\mem_{m, t}}{}
  }
  {m(c)}
$$
and
\[
\mathmakebox[\linewidth][c]{%
\left(
\begin{gathered}
  \nsem%
  {\st{\Tprogstar{c_m}}%
    {\vm_{m, t}}%
    {\mem_{m, t}}{}}
  {}{\seq{o_2}}
  {\st{\cnil}%
    {\vm_{m, t}}%
    {\mem'}{}}
  {n'}
  \\[0.5em]
  \text{or}\quad
  \nsem%
  {\st{\Tprogstar{c}}{\vm_{m, t}}{\mem_{m, t}}{}}
  {}{\seq{o_2}}
  {\sterror}
  {n'}
\end{gathered}
\right)
}
\]


We are going to invert two traces into source execution.
First, applying the step-wise simulation lemma (\cref{lemma:step-sim}), we can invert the first trace to a speculative source step:
\begin{align*}
\sem{\st{c}{\vm}{\mem}{\ms}}{\seq{d_1}}{T^{-1}(\seq{o_1})}{\st{c_{m}}{\vm_{m}}{\mem_{m}}{\ms_m}}.
\end{align*}
Then, by the induction hypothesis, the multi-step target trace above implies
\begin{gather*}
\exists c',
\sem*
{\st{c}{\vm_m}{\mem_m}{\ms_m}}
{\seq{d_2}}{T^{-1}(\seq{o_2})}
{\st{\cnil}{\vm'}{\mem'}{\ms'}}
\land
final(\st{\cnil}{\vm'}{\mem'}{\ms'}).
\end{gather*}
By the source rule \refrule{fig:source-semantics}{trans}, we can combine two source traces into one trace:
\begin{gather*}
\exists c',
\sem*{\st{c}{\vm}{\mem}{\ms}}{\seq{d_1} \lapp \seq{d_2}}{T^{-1}(\seq{o_1}) \lapp T^{-1}(\seq{o_2})}{\st{c'}{\vm'}{\mem'}{\ms'}}
\land
final(\st{c'}{\vm'}{\mem'}{\ms'}),
\end{gather*}
where $T^{-1}(\seq{o_1}) \lapp T^{-1}(\seq{o_2}) = T^{-1}(\seq{o_1} \lapp \seq{o_2})$ and $\seq{d_1} \lapp \seq{d_2} = \seq{d}$.
\end{proofcase}
\end{proofinductivecase}
\end{proof}

\section{An Auxiliary Transformation: Assert Elimination}%
\label{app:assert-elim}

Assert elimination models the instruction $\iassert{e}$ by the basic instructions in standard sequential semantics.
\Cref{fig:assert-elim} presents the assert elimination pass: $\Tass{\cdot}$ for the whole program and $\Tassstar{\cdot}$ for a code piece.
We introduce another ghost boolean variable $\retvar$, of which the boolean value indicates whether the program should stop.

The idea of the transformation is as follows.
$\retvar$ will be updated by $\iassert{e}$ and used to emulate halting the program.
The instruction, $\iassert{e}$, itself is replaced by a normal assignment: to set $\retvar$ as $\top$.
When $\retvar$ is set as $\top$, every instruction will be skipped and the loop will terminate as $\neg \retvar$ is conjoined to the loop condition.
Based on the main idea, we also make an optimization to avoid redundant checks:
we check only when $\retvar$ is possible to be freshly set as $\top$ according to the syntax.
For example, $\iift{\neg \retvar}{\iassign{x}{1}} \capp \iift{\neg \retvar}{\iassign{x}{2}}$ is optimized to $\iift{\neg \retvar}{\iassign{x}{1} \capp \iassign{x}{2}}$.

\begin{figure}
  \figasserttrans{}
  \caption{Assert elimination.}%
  \label{fig:assert-elim}
\end{figure}


The soundness and completeness of assert elimination is as follows.
Intuitively, this soundness and completeness says that the execution results before and after the transformation are equal up to the ghost variable $\retvar$.

\begin{lemma}[Soundness and completeness of assert elimination]%
\label{thm:correct-assert-elim}
    For any program $c$ and any input $i$,
    $$
    \sem*%
      {c(i)}
      {}{}
      {\st{\cnil}{\vm}{\mem}{}}
    $$
    iff
    $$
    \sem*%
      {\Tass{c}(i)}
      {}{}
      {\st{\cnil}{\vm'}{\mem}{}},
    $$
    where $\vm = \vm' \setminus \{ \retvar \}$, i.e., $\vm'$ equals $\vm$ up to $\retvar$.
\end{lemma}

\begin{proof}[Proof Sketch]
We can prove a slightly more general statement, by an induction on the number of execution steps:
$$
\sem*
  {\st{c}{\rho}{\mu}{}}
  {}{}
  {\st{\cnil}{\rho'}{\mu}{}}
$$
iff
$$
\sem*
  {\st{\Tassstar{c}}{\msetaux{\rho}{\retvar \mapsto \bot}}{\mu}{}}
  {}{}
  {\st{\cnil}{\rho_r'}{\mu}{}},
$$
where $\rho' = \rho_r' \setminus \{ \retvar \}$, i.e., $\rho'$ is equal to $\rho_r'$ except for $\retvar$.

Combining this statement with the first execution step for initialization,
$$\sem{\st{\Tass{c}}{\rho}{\mu}{}}{}{}{\st{\Tassstar{c}}{\rho[\retvar \mapsto \bot]}{\mu}{}}$$
yields the lemma statement.
\end{proof}


  \fi
{}

\end{document}
